\documentclass[journal,11pt]{IEEEtran}
\usepackage[english]{babel}

\usepackage{algorithmic}
\usepackage[ruled,vlined]{algorithm2e}
\usepackage{amsfonts}
\usepackage{amsmath}
\usepackage{amsthm}
\usepackage{amssymb}
\usepackage{bbm}
\usepackage{cite}
\usepackage{color}
\usepackage[mathscr]{euscript}
\usepackage{enumitem}
\usepackage{epstopdf}
\usepackage{float}
\usepackage[T1]{fontenc}
\usepackage{glossaries}
\usepackage{graphicx}
\usepackage[latin9]{inputenc}
\usepackage{times}
\usepackage{latexsym}
\usepackage{mathtools}
\usepackage{multirow}
\usepackage{makeidx}
\usepackage{mathrsfs} 
\usepackage{nomencl}
\usepackage{rotating}
\usepackage{textcomp}
\usepackage{varioref}
\usepackage{xcolor}
\usepackage{phonetic}

\ifCLASSOPTIONcompsoc
\usepackage[caption=false,font=normalsize,labelfon
t=sf,textfont=sf]{subfig}
\else
\usepackage[caption=false,font=footnotesize]{subfi
g}
\fi

\DeclareGraphicsExtensions{.eps,.png}

\allowdisplaybreaks
\newcommand{\bpara}[1]		{\bigskip \noindent {\bf #1}}

\usepackage{geometry}
\usepackage{soul,hyperref}
\usepackage[scaled=.80]{beramono}

\usepackage{cite,setspace}

\hypersetup{
    colorlinks=true,
    linkcolor=blue,
    filecolor=magenta,      
    urlcolor=blue,
    citecolor = red,
    }

\usepackage{xcolor,tikz}
\definecolor{amber}{rgb}{1.0, 0.49, 0.0}
\definecolor{cadmiumgreen}{rgb}{0.0, 0.42, 0.24}
\usetikzlibrary{calc}

\newtheoremstyle{styleth}%
{6pt}% Space above
{6pt}% Space below 
{}% Body font
{}% Indent amount
{\bfseries\color{black}}% Theorem head font
{}% Punctuation after theorem head
{.5em}% Space after theorem head
{}% Theorem head spec (can be left empty, meaning ‘normal’)
\theoremstyle{styleth}
\newtheorem{theo}{Theorem}

\newtheoremstyle{styledef}%
{5pt}% Space above
{5pt}% Space below 
{}% Body font
{}% Indent amount
{\bfseries\color{black}}% Theorem head font
{}% Punctuation after theorem head
{.2em}% Space after theorem head
{}% Theorem head spec (can be left empty, meaning ‘normal’)
\theoremstyle{styledef}

\newtheorem{definition}{Definition}
\newtheorem{prop}{Proposition}
\newtheorem{lem}{Lemma}

\newcommand{\statedefsolid}[2]{
  \par\noindent\tikzstyle{mybox} = [fill=yellow!10,
   thick,rectangle,inner sep=10pt]
  \begin{tikzpicture}
   \node [mybox] (box){%
    \begin{minipage}{#1}{#2}\end{minipage}
   };
  \end{tikzpicture}
}

\newcommand{\statetheoremsolid}[2]{
  \par\noindent\tikzstyle{mybox} = [draw=red,fill=gray!5,
   thick,rectangle,rounded corners,inner sep=6pt]
  \begin{tikzpicture}
    \node [mybox] (box){%
    \begin{minipage}{#1}{#2}\end{minipage}
   };
  \end{tikzpicture}
}

\newcommand{\statelemsolid}[2]{
  \par\noindent\tikzstyle{mybox} = [draw=blue,fill=blue!10,
   thick,rectangle,rounded corners,inner sep=6pt]
  \begin{tikzpicture}
    \node [mybox] (box){%
    \begin{minipage}{#1}{#2}\end{minipage}
   };
  \end{tikzpicture}
}

\newcommand{\statepropsolid}[2]{
  \par\noindent\tikzstyle{mybox} = [fill=red!5,
   thick,rectangle,inner sep=6pt]
  \begin{tikzpicture}
   \node [mybox] (box){%
    \begin{minipage}{#1}{#2}\end{minipage}
   };
  \end{tikzpicture}
}

\ifCLASSINFOpdf
\else
\fi

\onecolumn
\begin{document}
\renewcommand{\thefootnote}{\normalsize \arabic{footnote}} 	

\newcommand{\pw}{PW_{\Omega}}
\newcommand{\kz}{k\in\mathbb{Z}}
\newcommand{\akz}{\forall k \in \mathbb{Z}}
\newcommand{\tu}{\mathcal{T}_u}
\newcommand{\db}{\bar{\delta}}
\newcommand{\Z}{\mathbb{Z}}
\newcommand{\N}{\mathbb{N}}
\newcommand{\R}{\mathbb{R}}
\newcommand{\LR}{L^2(\mathbb{R})}
\newcommand{\LO}{L^2([-\Omega,\Omega])}
\newcommand{\re}{\mathbb{R}}
\newcommand{\co}{\mathbb{C}}
\newcommand{\cH}{\mathcal{H}}
\newcommand{\Tu}{\mathcal{T}_u}
\newcommand{\F}{\mathcal{F}}
\newcommand{\bs}{\boldsymbol}
\newcommand{\eq}{\triangleq}
\newcommand{\cond}{S}
\newcommand{\KN}{\mathcal{K}_N}
\newcommand{\Sz}{\mathbb{S}_N}
\newcommand{\s}{ISI}
\newcommand{\MO}[0]{\mathscr{M}_\lambda}
\newcommand{\fig}[1]{Fig.~\ref{#1}}
\newcommand{\MOh}[0]{\mathscr{M}_{\boldsymbol{\mathsf{H}}}}
\newcommand{\IFint}[0]{\mathcal{L}_k}
\newcommand{\IF}[0]{\mathrm{ASDM}_\theta}
\newcommand{\EQc}[1]		{\stackrel{(\ref{#1})}{=}}

\newcommand{\fe}[1]{\left[\kern-0.30em\left[#1  \right]\kern-0.30em\right]}

\newcommand{\FD}[3]{\mathscr{D}_{#1}^{#2}\left[#3\right]}
\newcommand{\NFD}{\mathscr{D}}
\newcommand{\syntharg}[1]{\mathcal{S}_\Omega\sqb{{#1}}}
\newcommand{\synth}{\mathcal{S}_\Omega}

\newcommand{\locav}{\mathcal{L}}

\newcommand{\diff}{d}

\newcommand{\FS}[3]{\mathscr{F}\rb{{#1}_{k}^{#3},#2_{k}}}

\newcommand{\vb}[1]{\left\lvert #1 \right\rvert}
\newcommand{\rb}[1]{\left( #1 \right)}
\newcommand{\sqb}[1]{\left[ #1 \right]}
\newcommand{\cb}[1]{\left\lbrace #1 \right\rbrace}
\newcommand{\floor}[1]{\left\lfloor #1 \right\rfloor}
\newcommand{\ceil}[1]{\left\lceil #1 \right\rceil}

\newcommand\ab[1]			{{\color{blue}#1}}

\newcommand\abb[1]			{{\color{red}#1}}
\newcommand\ETP[1]         {\mathbb{E}_{p}\left(#1\right)}

\newcommand\ETPP[2]         {\mathbb{E}_{#1}\left(#2\right)}

\renewcommand\tilde{\widetilde}

\DeclarePairedDelimiter{\norm}{\Vert}{\Vert}
\DeclarePairedDelimiter{\abs}{\left|}{\right|}
\DeclarePairedDelimiter{\Prod}{\langle}{\rangle}

\newcommand{\PW}[1]{\mathsf{PW}_{#1}}

\newcommand{\M}[3]{\vb{\mu_{#1}^{#2}\sqb{#3}}}
\newcommand{\B}[3]{\vb{\beta_{#1}^{#2}\sqb{#3}}}

\newcommand{\iPW}[2]{#1 \in \mathsf{PW}_{#2}}

\renewcommand\geq\geqslant \renewcommand\leq\leqslant
\newcommand{\const}[1]{#1}

\def\ind{{\color{red}{\pmb{1}}}}

\def\ind{\mathbbmtt{1}}

\def\th{\Psi}

% -- for turning on and off PH figures
\def\figmode    {1}         % DO NOT CHANGE THIS
\def\PH         {0}         % 1 FOR PLACEHOLDER, 0 FOR FIGURES
\def\stabilization	{0}
\def\True		{1}

\def\LemProofInAppendix {1}
\def\ThProofInAppendix  {1}
\def\PropProofInAppendix  {1}
\def\TR                 {1}
\def\FL                 {0}

\title{Time Encoding via Unlimited Sampling:\\ 
Theory, Algorithms and Hardware Validation}

\author{Dorian~Florescu 
	and~Ayush~Bhandari
	
\thanks{This work is supported by the UK Research and Innovation (UKRI) council's \emph{Future Leaders Fellowship} program ``Sensing Beyond Barriers'' (MRC Fellowship award no.~MR/S034897/1). Project page for (future) release of hardware design, code and data: \href{https://bit.ly/USF-Link}{\texttt{https://bit.ly/USF-Link}}.}	
\thanks{D.~Florescu and A.~Bhandari are with the Department of Electrical and Electronic Engineering, Imperial College London, SW72AZ, UK. E-mails: \{D.Florescu, A.Bhandari\}@imperial.ac.uk or ayush@alum.MIT.edu.}
\thanks{Manuscript received on Feb. XX, 2022; accepted with minor revisions on August XX, 2022.}
}

\markboth{Manuscript, August~2022}
{IEEE TSP}

\maketitle

\begin{abstract}
\normalsize
An alternative to conventional uniform sampling is that of time encoding, which converts continuous-time signals into streams of trigger times. This gives rise to Event-Driven Sampling (EDS) models. The data-driven nature of EDS acquisition is advantageous in terms of power consumption and time resolution and is inspired by the information representation in biological nervous systems. If an analog signal is outside a predefined dynamic range, then EDS generates a low density of trigger times, which in turn leads to recovery distortion due to aliasing. In this paper, inspired by the Unlimited Sensing Framework (USF), we propose a new EDS architecture that incorporates a modulo nonlinearity prior to acquisition that we refer to as the modulo EDS  {or MEDS}. In MEDS, the modulo nonlinearity folds high dynamic range inputs into low dynamic range amplitudes, thus avoiding recovery distortion. In particular, we consider the asynchronous sigma-delta modulator (ASDM), previously used for low power analog-to-digital conversion. 
This novel MEDS based acquisition is enabled by a recent generalization of the modulo nonlinearity called modulo-hysteresis. We design a mathematically guaranteed recovery algorithm for bandlimited inputs based on a sampling rate criterion and provide reconstruction error bounds. We go beyond numerical experiments and also provide a first hardware validation of our approach, thus bridging the gap  {between theory and} practice, while
corroborating the conceptual underpinnings of our work.  

\end{abstract}


\IEEEpeerreviewmaketitle

\newpage
\tableofcontents

\setstretch{1.2}
\section{Introduction}
\label{sec:intro}

In {Shannon's sampling paradigm, acquisition is performed by recording the amplitude of a signal at predefined, uniform time locations. Alternatively, in \emph{time encoding} the signal is converted into a nonuniform sequence of time events, leading to the acquisition paradigm called Event-Driven Sampling (EDS) \cite{Lazar:2004,Lazar:2008,Gontier:2014:J,MartinezNuevo:2019}.

Such time events are induced by \emph{significant} changes in the input signal values, thus enabling a data-driven approach to sampling.} EDS has been adopted in engineering fields such as control engineering and signal processing \cite{Miskowicz:2018} and has a wide range of applications, including neuromorphic vision \cite{Gallego:2020}, machine learning \cite{Maass:1997,Florescu:2019:J} and brain-machine interfaces (BMIs) \cite{Ozols:2012:C}. In biology, EDS is used to model the signal transmission in the nervous system of vertebrates \cite{Gerstner:2014}. When compared to uniform sampling, event-driven sampling schemes are not controlled by an external clock signal and are characterized by low power consumption \cite{Renaudin:2000:J,Wang:2021:C} and increased time resolution \cite{Gallego:2020}.

\begin{figure}[!h]
	\centerline{\includegraphics[trim={0cm 0cm 0cm 0},clip,width=10.5cm]{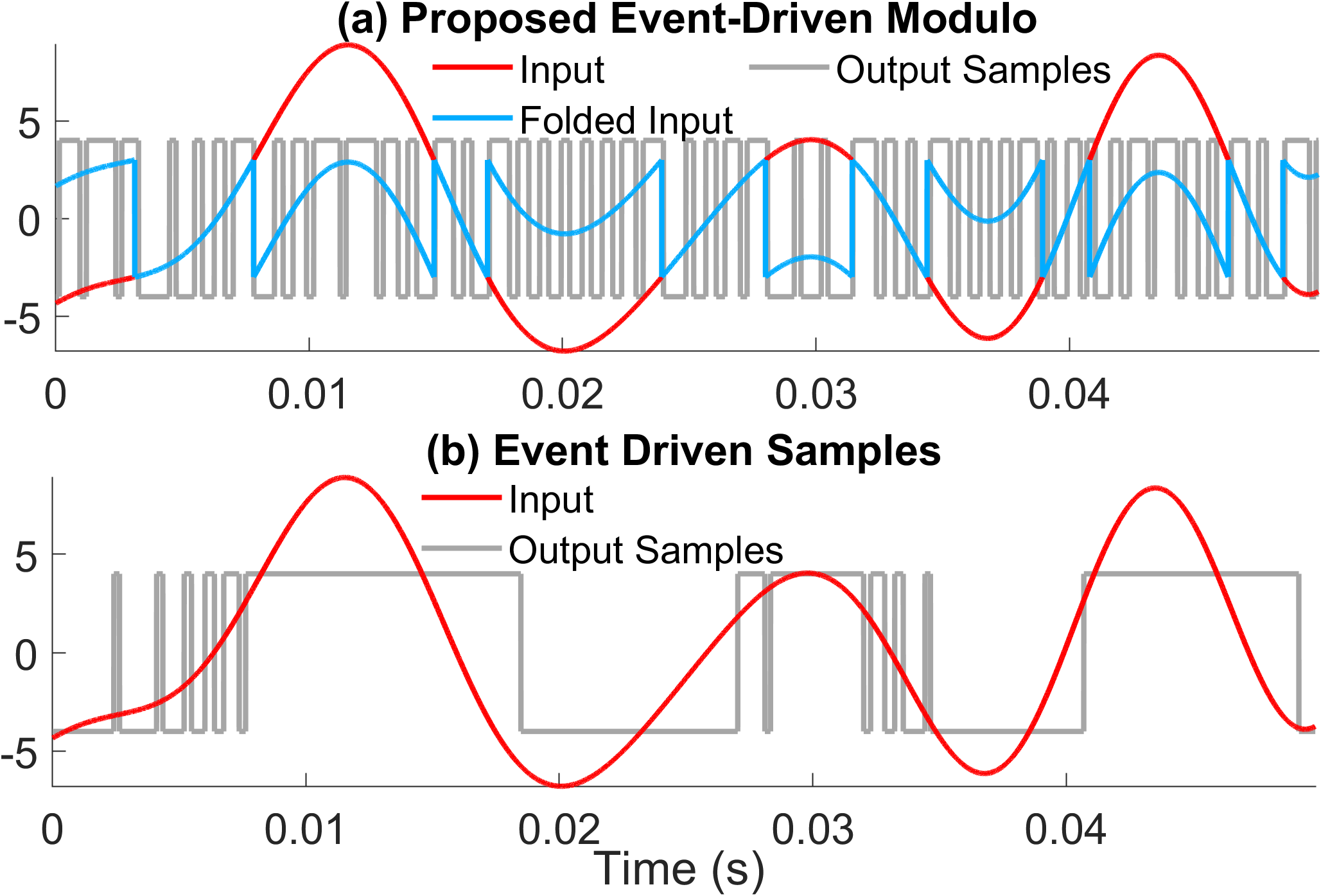}}
	\caption{The output of the proposed modulo event-driven sampling (MEDS) model and the standalone EDS implemented as an asynchronous sigma-delta modulator (ASDM) for a high dynamic range input. (a) The MEDS model folds the input amplitude thus preventing saturation. (b) The standalone ASDM has a suppressed output when the input amplitude is large. Our MEDS hardware prototype is depicted in \fig{fig:hardware}.}
	\label{fig:encoder}
\end{figure}

\bpara{Bottleneck in the Current EDS Acquisition.} In EDS, the input is encoded into an asynchronous sequence of time events or trigger times. The information content in such a sequence can be characterised by the sampling density, a positive quantity measuring the number of trigger instants per unit time. This leads to an encoding barrier: no information is being transmitted for inputs whose amplitudes are outside the range where the sampling density is positive.  {The} barrier is illustrated for the particular case of the asynchronous sigma-delta modulator  {based} EDS in \fig{fig:encoder}.  {Therefore,} the accurate recovery of the input is affected, which requires a minimal sampling density to prevent aliasing \cite{Feichtinger:1994}.  Thus, in  {EDS}, the input is typically constrained to a predefined dynamic range.

The problem of limited dynamic range is common in traditional encoding schemes such as analog-to-digital converters (ADCs), which  {clip or saturate the input} signals exceeding a predefined threshold. 
 {To overcome this fundamental bottleneck}, the Unlimited Sensing Framework (USF) \cite{Bhandari:2017,Bhandari:2019a,Bhandari:2020:J,Bhandari:2020:Pat,Bhandari:2021:J}  {was introduced and recently validated in hardware via modulo ADCs \cite{Bhandari:2021:J}---it was shown that, in practice, signals up to $24$ times larger than ADC's threshold can be recovered using modulo ADCs \cite{Bhandari:2021:J}, despite non-idealities and quantization effects}. In USF, the introduction of modulo nonlinearity converts a signal from high dynamic range (HDR) to low dynamic range. This is done by \emph{folding} the {HDR, continuous-time signal back into the dynamic range of the ADC, whenever the input amplitude} is outside a predefined range \cite{Bhandari:2017,Bhandari:2019a,Bhandari:2020:J,Bhandari:2020:Pat,Bhandari:2021:J}. The follow-up work on USF includes,
\begin{itemize}
  \item Different signal  {classes}, such as sparse signals \cite{Bhandari:2018,Bhandari:2022:J}, sinusoidal mixtures \cite{Bhandari:2018a} and spline spaces \cite{Bhandari:2020:C}.
    \item New sampling architectures \cite{Florescu:2021:J,Graf:2019,Florescu:2021:C,Florescu:2022:C}.
  \item Application frontiers such as array signal processing \cite{FernandezMenduina:2020:C,FernandezMenduina:2021:J}
and computational imaging \cite{Bhandari:2020:C,Bhandari:2020:Ca,Beckmann:2021:J}.
\end{itemize}

A more general model for modulo nonlinearity was introduced recently, namely,  \emph{modulo-hysteresis}\cite{Florescu:2021:C,Florescu:2021:J}, which, compared to ideal modulo, comprises two additional  {parameters that} can tackle non-idealities present in experimental observations. Moreover, modulo-hysteresis provided guarantees for a new class of reconstruction methods based on thresholding.

To address the dynamic range limitation in event-driven sampling schemes, we propose a new encoding model consisting of a modulo-hysteresis nonlinearity in series with an EDS model, which we call modulo EDS or MEDS. In our setting the EDS model is fixed, and the EDS input (modulo output) is guaranteed to satisfy the EDS dynamic range by selecting an appropriate modulo nonlinearity (\fig{fig:encoder}(a)). None of the existing recovery methods for EDS models can address the reconstruction of the modulo output, which is a folded signal. Therefore,  {the use of advanced recovery algorithms enables high dynamic range acquisition via MEDS that is conventionally not possible.}

The EDS model in this paper is an asynchronous sigma-delta modulator (ASDM) \cite{Lazar:2004}. Thanks to its low power consumption \cite{Roza:1997} and modular design \cite{Ozols:2015:C}, the ASDM was  {first implemented as an alternative to conventional ADCs in} \cite{Wei:2006:C} and  {subsequently} included in applications such as brain-machine interfaces \cite{Ozols:2012:C}. Assuming that the input amplitude is in a predefined range of values, the ASDM transmits trigger times measuring changes in the input integral \cite{Lazar:2004}. For inputs of amplitude outside this range, the ASDM transmits no information, as shown in \fig{fig:encoder}(b).

\bpara{Related Work.} The ASDM dynamic range can be extended by changing the model parameters or architecture. This includes increasing the sampling density by adjusting the triggering threshold \cite{Lazar:2004}, or introducing an adaptation mechanism via a new ASDM architecture \cite{Shavelis:2017:C}. Another popular EDS model is the integrate-and-fire (IF) neuron, inspired from the biological neuron, for which input reconstruction was studied for a wide range of input classes \cite{Gontier:2014:J,Alexandru:2019,Rudresh:2020} including  {recovery methods for EDS leveraging uniform sampling theory} \cite{Florescu:2015}. The recovery methods for IF also assume that the input is in a predefined dynamic range. 
It is possible to encode high dynamic range signals with an IF by generating bidirectional event-driven samples \cite{Alexandru:2019,Feichtinger:2012}. In this case, however, recovery guarantees do not exist for low amplitude inputs, for which the model does not generate output. 
Event based cameras, which use models closely related to the IF or ASDM, capture high dynamic range data with logarithmic pixels.  {This, in turn, decreases the pixel resolution for high intensities\cite{Gallego:2020}.} None of the methods above can recover modulo folded inputs from the output of an EDS. A MEDS architecture that extends the IF dynamic range was first numerically validated in \cite{Florescu:2021:C}.

\bpara{Contributions.} Our contributions are as follows,
\begin{enumerate}[leftmargin = 20pt, label = $\bullet$]
	\item We introduce a new event-driven modulo sampling scheme, comprising a modulo-hysteresis model in conjunction with an ASDM, which addresses the dynamic range restriction for an ASDM model without any parameter or architecture alteration. The new scheme is fully compatible with the existing ASDM methodology, but can also accommodate inputs that do not satisfy the dynamic range requirement.
	\item By exploiting the modulo-hysteresis architecture, we provide theoretical conditions for which the input can be recovered from the output trigger times.
	\item Extensive numerical experiments and a validation of the MEDS-based hardware prototype (cf.~\fig{fig:hardware}) show the validity of our approach.
\end{enumerate}

\bpara{Notation.} We use $\mathbb{Z}$ and $\mathbb{R}$ to denote the set of integers and real numbers, respectively, and $\mathbb{N}$ to denote the set of positive integers. Continuous functions are denoted as $f\rb{t}$ and discrete sequences as $f\sqb{k}$. The Hilbert space of square-integrable functions is denoted as $L^2\rb{\mathbb{R}}$. The norm in any space $H$ is denoted as $\norm{f}_H$. The derivative of order $\const{N}$ is denoted as $f^{\rb{\const{N}}}\rb{t}$ and, for sequences, the finite difference of order $\const{N}$ is $\Delta^\const{N} f \sqb{k}$, which is computed by applying recursively $\Delta^{N+1} f \sqb{k}= \Delta^N f\sqb{k+1}- \Delta^N f\sqb{k}$, where $\Delta^{0} f = f$.
The space of square integrable functions bandlimited to $\Omega$ is the Paley-Wiener space denoted $\PW{\Omega}$. The indicator function $\ind_{S}\rb{t}$ is $1$ for $t\in S$ and $0$ otherwise. The floor and ceiling functions are $\floor{\cdot}$ and $\ceil{\cdot}$, respectively, and $[\![ x ]\!]=x- \lfloor x \rfloor$ denotes the fractional part of $x$. Furthermore, let $ \mathrm{supp} \rb{f}$ denote the support of sequence $f\sqb{k}$, let $\vb{S}$ denote the cardinality of a set $S$, and let $\emptyset$ denote the empty set. By $f_\infty$ or $\norm{f}_\infty$ we denote the absolute sequence norm $\norm{f}_{\ell^\infty}$ or the absolute function norm $\norm{f}_{L^\infty}$, respectively. Let $\mathrm{sinc}_\Omega$ be the sinc function defined here as $\mathrm{sinc}_\Omega\rb{t}\triangleq \frac{\sin \rb{\Omega t}}{\pi t}$.

\bpara{Scope and Organization.} The ASDM model and associated reconstruction methods are summarised in Section \ref{sect:ASDM}. The MEDS architecture is introduced in Section \ref{sect:new_encoder}. The new reconstruction method from MEDS samples and associated theoretical guarantees are given in Section \ref{sect:main_results}. The numerical demonstration and MEDS hardware prototype are in Section \ref{sect:NumericalDemo}. Section \ref{sect:Proofs} comprises proofs for some of the theoretical results in Section \ref{sect:main_results}, and the conclusions are in Section \ref{sect:conclusions}.

\section{The Asynchronous Sigma-Delta Modulator}
\label{sect:ASDM}
The ASDM is an event-driven sampling model consisting of a feedback loop containing an adder, an integrator, and a noninverting Schmitt trigger, as illustrated in \fig{fig:asdm}. We introduce the formal definition of the ASDM model in Section \ref{subsect:ASDMdesc} and summarise the recovery method for bandlimited inputs and its shortcomings in Section \ref{subsect:ASDMrec}.

\begin{figure}[!t]
	\centerline{\includegraphics[trim={0cm 0cm 0cm 0},clip,width=9cm]{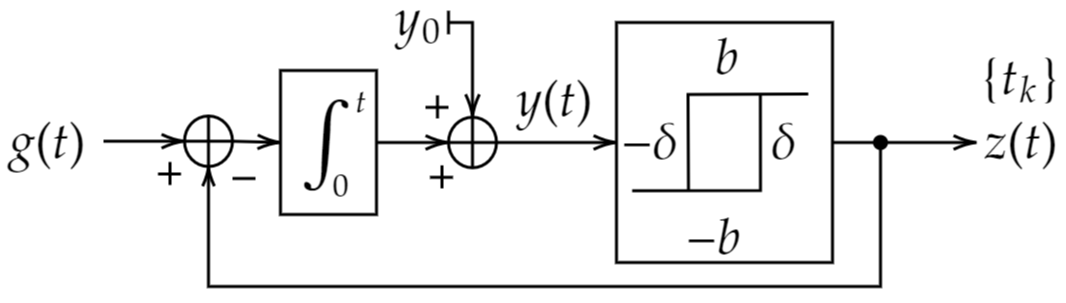}}
	\caption{The asynchronous sigma-delta modulator (ASDM) adds input $g\rb{t}$ with $z\rb{t}$, the current state of the ASDM, and the result is processed with an integrator with initial condition $y_0$ in series with a Schmitt trigger with parameters $\delta,b$.}
	\label{fig:asdm}

\end{figure}

\begin{figure}[!t]
\centerline{\includegraphics[trim={0cm 0cm 0cm 0},clip,width=9cm]{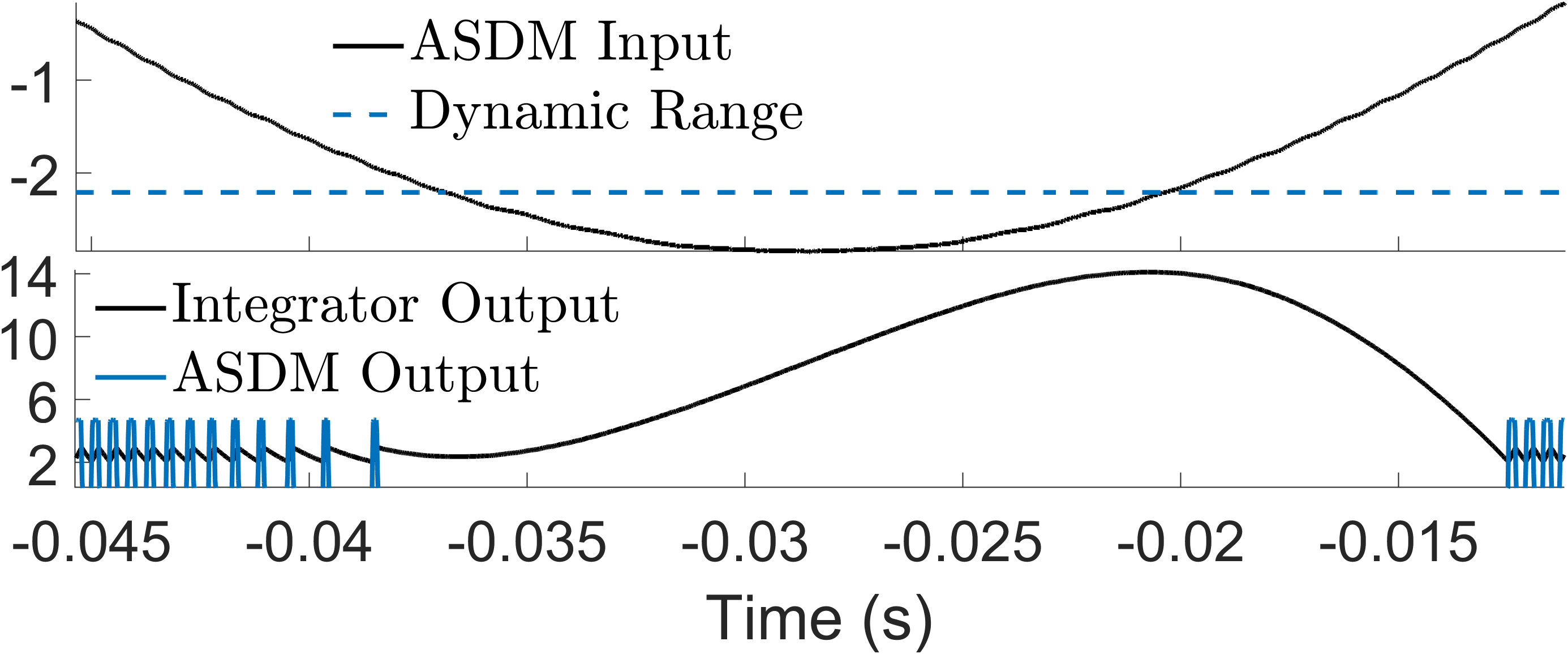}}
\caption{The effect of ASDM saturation observed in hardware measurements (cf.~ \fig{fig:hardware}(b)). Input $g(t)$ goes below the ASDM dynamic range, thus the integrator output $y(t)$ is no longer monotonic. The next trigger time is only generated when $g(t)$ is back in the dynamic range which pushes $y(t)$ to cross threshold $-b$. This illustration is based on the experiment reported in \cite{Florescu:2022:Ca}.}
\label{fig:saturation}
\end{figure}

\subsection{Model Description}
\label{subsect:ASDMdesc}

At the initial time point $t {=t_0}=0$, the output of the Schmitt trigger is set as $z\rb{t}=-b$.  {The signal $z\rb{t}$ retains this state until the first trigger time occurs, as will be explained next.} In a vicinity of the $t=0$ the integrator output $y(t)$ satisfies:
\begin{equation}
	y\rb{t}=y_0+\int_0^{t} \rb{g\rb{s} + b} ds,
	\label{eq:ASDM_init}
\end{equation}
where $g\rb{t}$ is the input,  {$y_0$ is a constant denoting the  {integrator's} initial condition,} and $b,\delta$ are the Schmitt trigger parameters. The ASDM generates output only if $|g\rb{t}|\leq c<b$, which implies that $y\rb{t}$ is strictly increasing \eqref{eq:ASDM_init}.  {When $|g\rb{t}|> c$, $y\rb{t}$ is no longer monotonic, and thus may not cross the threshold $\delta$. This saturates the ASDM, preventing any triggers until the input is back in the ASDM dynamic range (see \fig{fig:saturation}).} The trigger time $t_1$ corresponds to $y\rb{t_1}=\delta$, which flips the output of the Schmitt trigger to $b$. This enforces $y\rb{t}$ to be strictly decreasing, and the next trigger time satisfies $y\rb{t_2}=-\delta$.
The sequence of trigger times $\cb{t_k}$ is directly linked with the input and ASDM parameters via the following recursive equation, also known as the \emph{t-transform} \cite{Lazar:2004}
\begin{equation}
\int_{t_k}^{t_{k+1}} \rb{g\rb{s} + \rb{-1}^k b} ds=\rb{-1}^k 2\delta,\quad  {k\in\Z_+^*}.
\label{eq:ttransform}
\end{equation}

We model the link between the trigger times and input via the operator $\cb{t_k}=\mathrm{ASDM}_{\delta,b} \rb{g}$.
The $t$-transform was also studied in the context of other EDS models such as the IF neuron \cite{Lazar:2008}. By defining $G\rb{t}\triangleq \int_0^t g\rb{s}ds$, it can be shown via induction that $G\rb{t_k}$ can be derived directly from \eqref{eq:ttransform}. This creates a strong link between ASDM sampling and nonuniform sampling, and allows using similar analysis and recovery approaches \cite{Feichtinger:1994}.  {A hardware realisation of an ASDM is in \fig{fig:hardware}(b).}

\subsection{Reconstruction from Asynchronous Sigma-Delta Samples}
\label{subsect:ASDMrec}

The ASDM input $ g(t) $ can be perfectly reconstructed \cite{Lazar:2008} if it is bandlimited $ g \in \PW{\Omega} $ and the maximum distance between the ASDM trigger times satisfies the Nyquist rate condition, i.e., $$ \vb{t_{k+1}-t_k}<\frac{\pi}{\Omega}. $$ This condition is sufficient to guarantee recovery from nonuniform samples $G\rb{t_k}$ \cite{Feichtinger:1994}. 

In the case of the ASDM, if $\vb{g(t)}>b$ the encoder does not generate any output, meaning that $\vb{t_{k+1}-t_k}<\frac{\pi}{\Omega}$ does not always hold. To overcome this problem it is assumed that $g_\infty < b$, which {, via \eqref{eq:ttransform},} guarantees that $t_{k+1}-t_k<\frac{2\delta}{b-g_\infty}$. Therefore, the Nyquist rate condition for recovery is guaranteed if
\begin{equation}
\label{eq:IFdynrange}
g_\infty<b-\frac{2\delta \Omega}{\pi}.
\end{equation}
The input is recovered as \cite{Lazar:2004}
\begin{equation}
\label{eq:rec_local_av}
    g\rb{t}=\sum\limits_{m\in\mathbb{Z}} c_m \mathrm{sinc}_\Omega  \rb{t-\bar{t}_m},
\end{equation}
where $\bar{t}_m=\rb{t_m+t_{m+1}}/{2}$ are the midpoints of the intervals between the trigger times. The coefficients $c_m$ are computed via least squares from $$\int_{t_k}^{t_{k+1}} g\rb{s} ds \EQc{eq:ttransform}\rb{-1}^k 2\delta-\rb{-1}^k b\rb{t_{k+1}-t_k}.$$ In the context of nonuniform sampling, this procedure is also known as {\it reconstruction from local averages} \cite{Feichtinger:1994}.

Intuitively, whenever $\vb{g\rb{t}}>b$, the input does not trigger the ASDM (\fig{fig:encoder}(b)). This is because the integrator output $y\rb{t}$ is no longer monotonic and might not reach the threshold $\pm b$. This locks in the ASDM output until the threshold is reached. This effect is also observed in hardware, as illustrated in \fig{fig:saturation}. As will be shown in Section \ref{sect:NumericalDemo}, the ASDM saturation leads to unstable reconstructions.

\begin{figure}[!t]
	\centerline{\includegraphics[trim={0cm 0cm 0cm 0},clip,width=12cm]{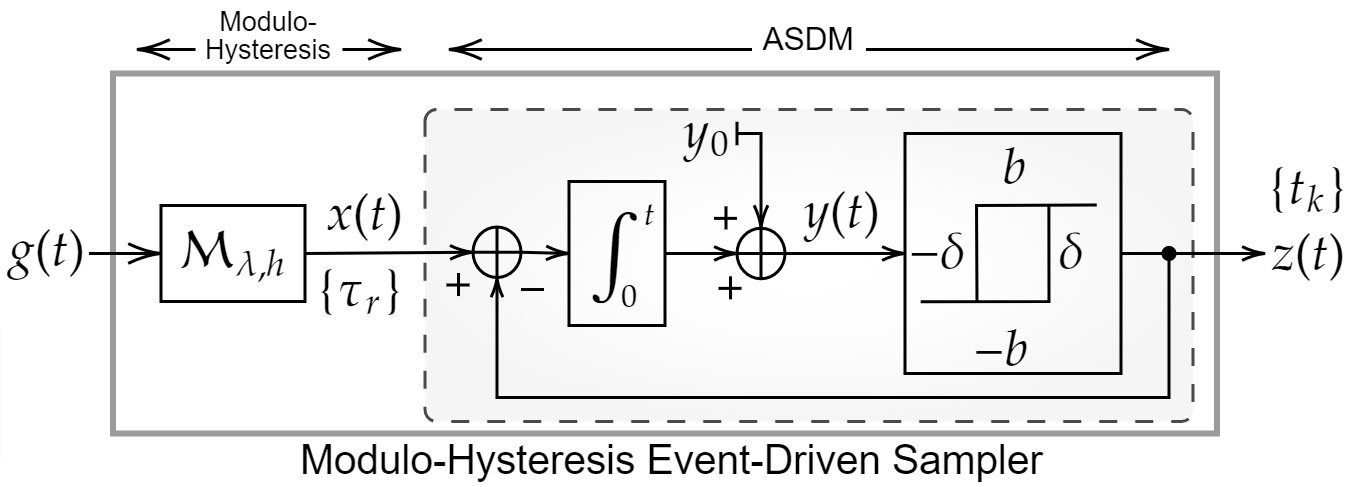}}
	\caption{The proposed encoding model.}
	\label{fig:diagram}
\end{figure}

\section{The Proposed Modulo Event-Driven Encoder}
\label{sect:new_encoder}
Here we introduce a new event-driven modulo encoder model, comprising a modulo encoder with hysteresis and an ASDM model,  {depicted in \fig{fig:diagram}}. We first present the ideal modulo encoder in \ref{subsect:idealmodulo} and explain some of its limitations. The modulo-hysteresis is then described in \ref{subsect:hystmodulo}. The proposed event-driven modulo is introduced in \ref{subsect:new_encoder}.

\subsection{The Ideal Modulo}
\label{subsect:idealmodulo}
Let $g \in \PW{\Omega} $ where $\PW{\Omega}$ is the Paley-Wiener space of $\Omega$--bandlimited functions. As in \cite{Bhandari:2018}, the modulo nonlinearity is defined by function $\MO:\PW{\Omega}\rightarrow L^2\left(\mathbb{R}\right)$
\begin{equation}
\MO \left(g\rb{t}\right) = 2\lambda \left( {\fe{ {\frac{g\rb{t}}{{2\lambda }} + \frac{1}{2} } } - \frac{1}{2} } \right)
\label{eq:fold}
\end{equation}
where $[\![ x ]\!]=x- \lfloor x \rfloor$ is the fractional part of $x$, and $L^2\rb{\mathbb{R}}$ denotes the space of functions with finite energy. 
The instantaneous times when $[\![ x ]\!]$ switches from $1$ to $0$ are called are called \emph{folding times} and denoted $\tau_r, r \in \mathbb{Z}$. The output of the model $x\rb{t}=\MO \rb{g\rb{t}}$ can be expanded as, $$x\rb{t}=g\rb{t}-\varepsilon_g\rb{t} \mbox{ where } \varepsilon_g\rb{t}=2\lambda\sum\limits_{r\in\mathbb{Z}}s_r \ind_{[\tau_r,\infty)}\rb{t},t\in\mathbb{R}, \quad s_r=\pm1.$$
In Unlimited Sampling the data is sampled uniformly with period $T$, which yields
\begin{equation}
\label{eq:ideal_mod_data}    x\rb{kT}=g\rb{kT}-\varepsilon_g\rb{kT},\quad k \in\mathbb{Z}.
\end{equation}
The input $g\rb{kT}$ is recovered by applying a finite difference filter $\Delta^N$ to \eqref{eq:ideal_mod_data}, which annihilates $\Delta^N\varepsilon_g\rb{kT}$  {given} $\Delta^N x\rb{kT}$. This is possible because $\Delta^N$ vanishes gradually $\Delta^N g\rb{kT}$ but always leaves $\Delta^N \varepsilon_g\rb{kT}$ on an equally spaced grid that is then annihilated via a modulo operation. However this approach does not work for experimental modulo data due to various non-idealities \cite{Bhandari:2021:J,Florescu:2021:J}. This  {approach also would not work} when sampling the output with an ASDM, which integrates the input on intervals of various durations. \eqref{eq:ttransform}.

\begin{figure*}[!t]
\includegraphics[width=1\textwidth]{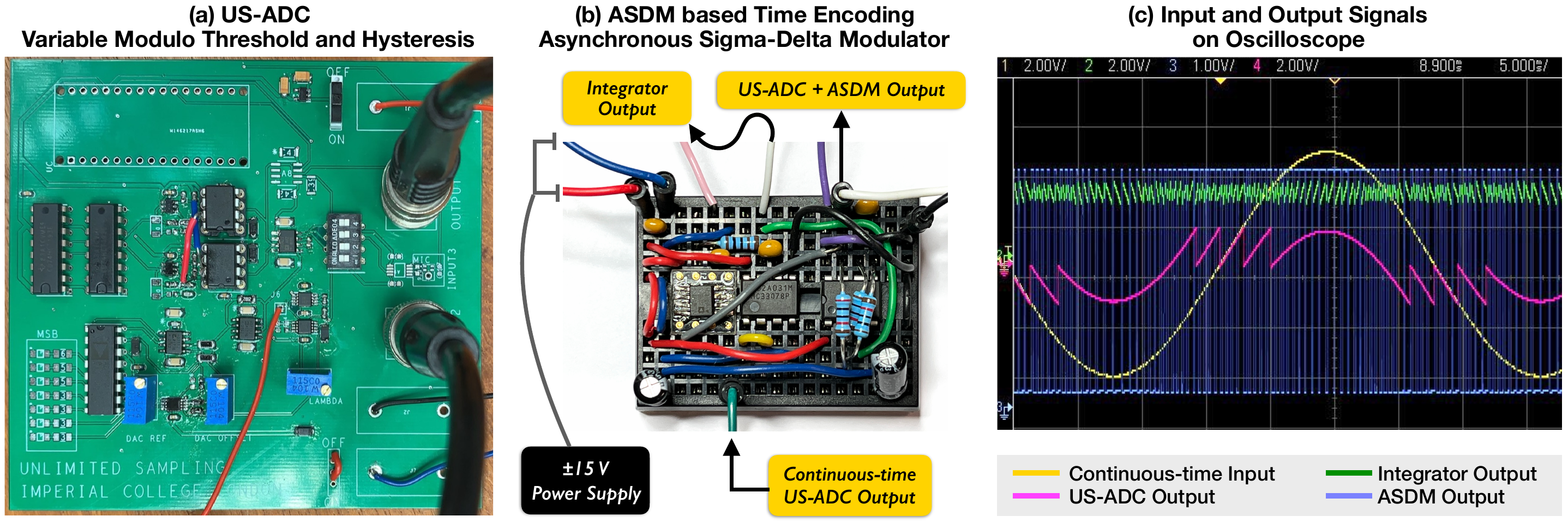}
\caption{Hardware experiment. (a) The US-ADC hardware prototype with variable parameters controlling the threshold $\lambda$ and hysteresis $h$. (b) The ASDM hardware realisation. (c) Oscilloscope screenshot showing the input and outputs of the modulo-hysteresis and ASDM hardware prototypes.}	
\label{fig:hardware}
\end{figure*}

\begin{figure}[!t]
    \centerline{\includegraphics[trim={1cm 0cm 1.5cm 0},clip,width=0.75\textwidth]{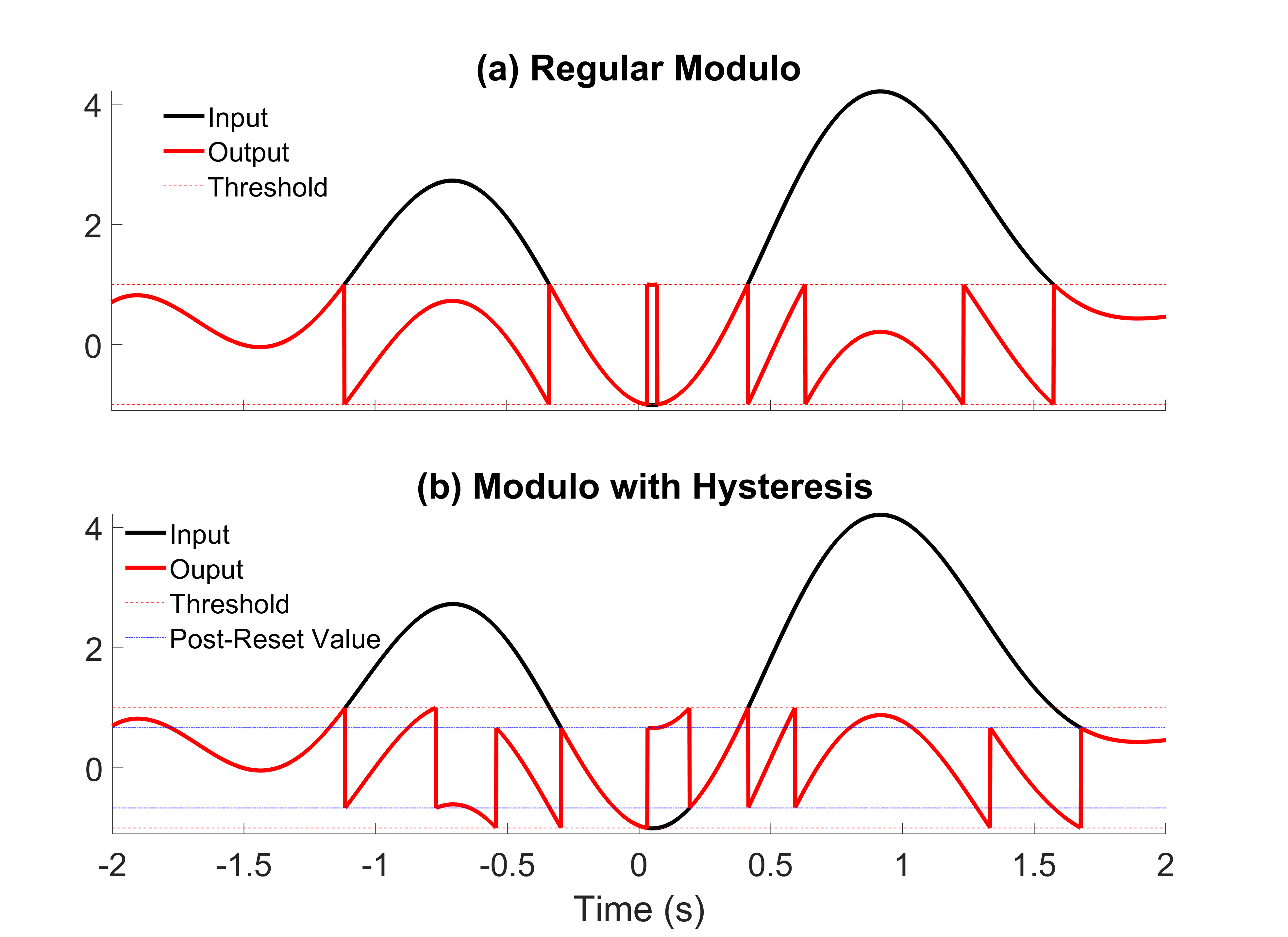} }
    \caption{Encoding a random bandlimited waveform with two types of modulo nonlinearities. (a) The ideal modulo resets the input to the opposite threshold value, and can generate arbitrarily close folding times \cite{Bhandari:2017,Bhandari:2018}. (b) The  modulo-hysteresis nonlinearity separates the thresholds and post-reset values to guarantee a minimal separation between the folding times.}
    \label{hysteresis}
\end{figure}

A different reconstruction approach from uniform samples for the ideal modulo consists of filtering the modulo output with a kernel that amplifies the folding times, which are subsequently detected via \emph{thresholding} \cite{Rudresh:2018,Graf:2019}.
The thresholding approach is particularly desirable when the data is affected by non-idealities, such as transient folding samples, where the Unlimited Sampling approach fails \cite{Florescu:2021:J}. However, the input reconstruction is not guaranteed when the folding times are very close, which can happen for the ideal modulo, as depicted in  \fig{hysteresis}(a). This problem can be solved using the modulo-hysteresis model \cite{Florescu:2021:J}, presented in the next subsection, which will be used in the present work.

\subsection{The Modulo-Hysteresis Model}
\label{subsect:hystmodulo}
The phenomenon of \emph{hysteresis}, present in ADC circuits \cite{Pereira:1999:C}, results from a difference between the quantization threshold values when the value is approached from above or below \cite{IEEEstd:2017}. This effect is also commonly exploited in circuits such as the Schmitt trigger, which avoids noise induced spurious triggering during ADC conversion. Recently,  {that} effect was exploited in a modulo-hysteresis model and associated hardware implementation  \cite{Florescu:2021:J}  {see \fig{fig:hardware}(a)}, which limits the minimum distance between two consecutive folding times.
We give below the formal definition of modulo-hysteresis.

\statedefsolid{0.95\textwidth}{
\begin{definition}[Modulo-Hysteresis]
The modulo-hysteresis with threshold $\lambda$ and hysteresis parameter $h\in\left[0,\lambda\right)$, denoted $\MOh$, $\boldsymbol{\mathsf{H}}=\sqb{\lambda\ h}$, generates an nonuniform sequence $\cb{\tau_r}$ and an analog function $x \rb{t}=\MOh g \rb{t},$ $ t\geq \tau_0$ in response to input  $g \in \PW{\Omega}$. The sequence $\cb{\tau_r}$ is computed recursively as
\begin{align*}
\tau_1&=\min\cb{t > \tau_0 \vert \MO\rb{g\rb{t}+\lambda}= 0}\\
\tau_{r+1}&=\min\cb{t> \tau_r \vert \MO\rb{g\rb{t}-g\rb{\tau_r}+h s_r}=0},
\end{align*}
where $s_r= \mathrm{sign} \rb{g\rb{\tau_r}-g\rb{\tau_{r-1}}}, n\geq1$.
The function $x \rb{t}$ returned by the encoder is then given by  
\begin{equation}
\label{eq:zt}
x \rb{t}=g\rb{t}-\varepsilon_g\rb{t},
\end{equation}
where $\varepsilon_g\rb{t}=2\lambda_h\sum\nolimits_{r\in\mathbb{Z}}s_r \ind_{[\tau_r,\infty)}\rb{t}$ and  $\lambda_h\triangleq \lambda-h/2$.
\label{def:cont_modulo}
\end{definition}
}

Modulo-hysteresis has the following important property.
\begin{lem}[Minimum Separation between Folds]
\label{lem:separation}
Let $\cb{\tau_r}_{r\in\mathbb{Z}}$ be the sequence of folding times corresponding to encoder $\MOh$, $\boldsymbol{\mathsf{H}}=\sqb{\lambda\ h}$, for an input $g\in\PW{\Omega}$. Then
\begin{equation}
\label{eq:deltatau}
\tau_{r+1}-\tau_r \geq \frac{h^*}{\Omega g_\infty}, \quad h^*=\min\cb{h,2\lambda_h}.
\end{equation}
\end{lem}
\begin{proof}
From Bern\v{s}te\u{\i}n's inequality  $\norm{g^{\rb{n}}}_p\leq\Omega^n \norm{g}_p$  {(cf. pg. 116, \cite{Nikol:2012:B})} and Definition \ref{def:cont_modulo}, for $p=\infty$,  { $n=1$ and $r\geq1$}, 
\begin{equation*}
\tau_{r+1}-\tau_r  \geq \frac{\vb{g\rb{\tau_{r+1}}-g\rb{\tau_{r}}}}{\max \vb{g'\rb{t}}}  \geq \frac{\min\cb{h,2\lambda-h}}{\Omega\max \vb{g\rb{t}}}=\frac{h^*}{\Omega g_\infty}.
\end{equation*}
This proves the result.
\end{proof}

When the data is sampled with period $T$, it was shown that here $g\rb{kT}$ can still be recovered from $x\rb{kT}=\MOh g\rb{kT}$ with USF, given that $\Delta^N \varepsilon_g \rb{kT}$ is on a uniform grid with spacing $2\lambda-h$ \cite{Florescu:2021:J}. This property no longer holds when the data is sampled with an asynchronous sigma-delta modulator, as will be explained next.

\subsection{The Proposed Modulo Event-Driven Sampling Model}
\label{subsect:new_encoder}
The proposed model consists of a modulo-hysteresis in series with an ASDM model. Then,  {for $k\geq0$}, the trigger times $\cb{t_k}$ generated by the model satisfy \eqref{eq:ttransform}
\begin{equation}
\begin{gathered}
\int_{t_k}^{t_{k+1}} \MOh
g\rb{s}  ds=q_k,\\
q_k\triangleq \rb{-1}^k 2\delta-\rb{-1}^k b\rb{t_{k+1}-t_k}.
\end{gathered}
\label{eq:mod-IF}
\end{equation}
We assume $g\in\PW{\Omega}$ and $\lambda\leq b-\frac{2\delta \Omega}{\pi}$ such that if $\vb{g\rb{t}}<\lambda$ then \eqref{eq:IFdynrange} is satisfied, and thus $g$ is left unchanged by the modulo nonlinearity. This ensures that \eqref{eq:mod-IF} is equivalent to \eqref{eq:ttransform}, and thus the proposed encoder is backwards compatible with the existing results on event-driven sampling. When $\vb{g\rb{t}} > b-\frac{2\delta \Omega}{\pi}$, the asynchronous sigma-delta modulator leads to information loss, which is prevented by the proposed circuit by folding the input back into the range $\sqb{-\lambda,\lambda}$. According to Definition \ref{def:cont_modulo}, $x  \rb{t} = \MOh g\rb{t}$ satisfies
\begin{equation}
\label{eq:expanded_z}
x \rb{t}\EQc{eq:zt} 
g\rb{t}-2\lambda_h \sum\limits_{r\in\Z} s_r \ind_{\left[ \tau_r,\infty \right)}\rb{t}.
\end{equation}
Let $X\rb{t}=\int_{ {0}}^t x\rb{s} ds$. Then $ {X\rb{t_0}=X\rb{0}=0}$ and \eqref{eq:mod-IF}  {implies that}
\begin{equation}
\label{eq:nonunif_sampling}
    X\rb{t_k}=\sum\limits_{m= {0}}^{k-1} q_m,\quad k\geq {1}.
\end{equation}
Furthermore, let $ E_g\rb{t}\triangleq \int_0^t \varepsilon_g\rb{s}ds$ and $ G\rb{t}\triangleq \int_0^t g\rb{s}ds$, and thus, by integrating \eqref{eq:expanded_z} we get 
\begin{equation}
\label{eq:expanded_Z}
    X\rb{t_k}=G\rb{t_k}-E_g\rb{t_k}.
\end{equation}
 {We can then compute samples $X\rb{t_k}$ using just the model parameters and output trigger times via \eqref{eq:nonunif_sampling}. Further on, the input and  {residue} will be separated using \eqref{eq:expanded_Z}}.
Equation \eqref{eq:expanded_Z} has a similar form as \eqref{eq:ideal_mod_data}, but here $E_g\rb{t_k}$ no longer lie on an uniform grid as it involves integration on intervals of variable length. Therefore, a different recovery strategy than in the case of USF is required, which is explained in the following.

\section{Recovery from Modulo Event-Driven Data}
\label{sect:main_results}

Our goal is to recover $g\rb{t}$ from \eqref{eq:mod-IF} given time instances $\cb{t_k}_{k\in\Z}$ and design parameters $\cb{\lambda,h,b,\delta}$. 
The proposed encoding circuit guarantees the asynchronous sigma-delta modulator input is in a predefined range. 
Generally, the EDS input $x\rb{t}$ is not bandlimited, and does not even belong to a class of shift-invariant spaces, given that $ {\tau_r\geq0}$ do not lie on a uniform grid. Thus $x\rb{t}$ cannot be recovered with the traditional methods for EDS models, which assume that the input is bandlimited, or at least smooth \cite{Lazar:2008,Lazar:2010:J,Gontier:2014:J,Florescu:2021:J}. 

Given samples $X\rb{t_k}=G\rb{t_k}-E\rb{t_k}$, the main challenge is to recover both $G\rb{t_k}$ and $E\rb{t_k}$. By applying a generic filter $\NFD$ to the data $X\rb{t_k}$  {for $k\geq 2$} we get (\ref{eq:zt}, \ref{eq:nonunif_sampling})
\begin{equation}
\label{eq:Gtk}
\NFD X\rb{t_k}=\NFD G\rb{t_k}-\NFD E_g\rb{t_k}=\NFD \rb{\sum\limits_{m= {0}}^{k-1} q_m}.
\end{equation}

The input reconstruction is organised in two stages.
 {In the first stage, $\NFD G\rb{t_k}$ is interpreted as noise and an estimation $\widetilde{E}_g\rb{t_k}$ of $E_g\rb{t_k}$ is computed by thresholding $\NFD X\rb{t_k}$. This approach was used before to reconstruct modulo folded data in the context of uniform samples \cite{Rudresh:2018,Graf:2019, Florescu:2021:J}. In the second recovery stage, an estimation $\widetilde{G}\rb{t_k}$ of $G\rb{t_k}$ is computed from the integral form of the modulo-hysteresis input $X\rb{t_k}$ and $\widetilde{E}_g\rb{t_k}$ through \eqref{eq:expanded_Z}, and the respective estimate for the continuous input $\widetilde{g}\rb{t}$ is reconstructed from the estimated samples $\widetilde{G}\rb{t_k}$, such that the error norm $\norm{g-\widetilde{g}}_{L^2}$ is bounded.}

The section is organised as follows. To accommodate nonuniform samples, we define the nonuniform difference operator $\NFD^N$ in Section \ref{subsect:NonDiff}. We then describe the first recovery stage, and provide guarantees for computing $E_g\rb{t_k}$ in the particular case of one folding time in Section \ref{subsect:OneFold}, and extend it to multiple folding times in Section \ref{subsect:MultipleFolds}. Section \ref{subsect:InputRec} describes the second recovery stage and gives the input reconstruction algorithm and associated error bound.

\subsection{The Nonuniform Difference Operator}
\label{subsect:NonDiff}

Here we formulate a thresholding based recovery for the case of the ASDM model. Given samples $X\rb{t_k}=G\rb{t_k}-E\rb{t_k}$, the main challenge is to recover both $G\rb{t_k}$ and $E\rb{t_k}$. We threshold samples $\NFD X\rb{t_k}$ to recover the residual $E_g\rb{t}$, which requires identifying the folding times $\tau_r$ and signs $s_r$ \eqref{eq:expanded_z}.
Operator $\NFD$ is designed to enhance the effect of $E_g$ in \eqref{eq:Gtk}. In \cite{Florescu:2021:J} the data was processed with the uniform finite difference filter $\Delta^N$. Here we recover the input from nonuniform samples and, to this end, we choose $\NFD=\NFD^N$ defined recursively the nonuniform difference operator of order $N$, 
\begin{equation}
\label{eq:NFD}    
\NFD^N f\sqb{k}= \frac{\NFD^{N-1} f\sqb{k+1}-\NFD^{N-1} f\sqb{k}}{t_{k+N}-t_k},
\quad N>0,     
\end{equation}
where $\NFD^0 f\sqb{k} = f\sqb{k}$ and $f\sqb{k}$ is a generic sequence. Here, $\NFD^N$ 
is an extension of the difference operator $\Delta^N$ used for uniform samples. In \eqref{eq:Gtk}, the values $\NFD^N \sum_{m= {0}}^{k-1} q_m$ can be computed directly from the measurements. As is the case for USF, here the measurements lead to a mixture of the samples $\NFD^N G\rb{t_k}$ and $\NFD^N E_g\rb{t_k}$. However, unlike the uniform sampling scenario, values $\NFD^N E_g\rb{t_k}$ do not lie on an uniform grid.  {We will show that the separation of $\NFD^N G\rb{t_k}$ and $\NFD^N E_g\rb{t_k}$ is still possible via a thresholding approach. To this end, we provide the following bound, which will be used in our recovery approach}
(proof in Section \ref{sect:Proofs}).

\statepropsolid{0.95\textwidth}{
\begin{prop}
\label{prop:boundDG}
Let $\NFD^N$ be the nonuniform finite difference operator defined in \eqref{eq:NFD}. Then
\begin{equation}
\vb{ \NFD^N G\rb{t_k}}\leq \frac{1}{\Omega N!}\rb{\frac{T_\mathsf{max}}{T_\mathsf{min}}\Omega e}^N  g_\infty,
\label{eq:boundG}
\end{equation}
where $T_\mathsf{min}, T_\mathsf{max}>0$ are two constants satisfying $T_\mathsf{min} \leq  t_{k+1}-t_k \leq T_\mathsf{max}$.
\end{prop}
}

This bound is an extension to nonuniform samples of the bound derived in \cite{Bhandari:2017} for $\vb{\Delta^N g\rb{kT}}$. Specifically, we note that in the limit case where $T_\mathsf{min}=T_\mathsf{max}=T$, and thus $t_k=kT$, we have that 
$$\NFD ^N G\rb{t_k} \EQc{eq:NFD} \frac{\Delta^N G\rb{kT}}{T^N N!}.$$ According to Bern\v{s}te\u{\i}n's we have $\norm{G'}_\infty\leq \Omega G_\infty \Rightarrow g_\infty \leq \Omega G_\infty $. Then, in this particular case, \eqref{eq:boundG} becomes $\vb{\Delta^N G\rb{kT}}\leq \rb{T \Omega e}^N G_\infty$, which is in accordance to the bound derived in \cite{Bhandari:2017,Bhandari:2020:J}.

\subsection{The Case of a Single Fold}
\label{subsect:OneFold}
For simplicity, let us first assume we have only one folding time $\tau_1$. Let $K_1$ be the index of the trigger time preceding the folding time, such that $\tau_1\in\left[ t_{K_1},t_{K_1+1} \right)$. The following result calculates the expression of $\NFD^N E_g\rb{t_k}$.

\statepropsolid{0.95\textwidth}{
\begin{prop}
\label{prop2:Eg_mu_beta}
 {For all $k\in\Z$,} sequence $\NFD^N E_g\rb{t_k}$ satisfies
\begin{equation}
\label{eq:mu_k_beta_k}
    \frac{\NFD^N E_g\rb{t_k}}{s_1\rb{2\lambda-h}}=\mu_{K_1}^N\sqb{k} p +\beta_{K_1}^N\sqb{k}\rb{1-p}, \forall N\geq 2,
\end{equation}
where $k\in\Z$, $p=\frac{t_{K_1+1}-\tau_1}{t_{K_1+1}-t_{K_1}}\in\sqb{0,1}$ and $\mu_{K_1}^N\sqb{k}$, $\beta_{K_1}^N\sqb{k}$ satisfy 
\begin{equation}
\mu_{K_1}^{N+1}\sqb{k}=\frac{\mu_{K_1}^N\sqb{{k+1}}-\mu_{K_1}^N\sqb{{k}}}{t_{k+N+1} -  t_{k}} \mbox{ and } 
\beta_{K_1}^{N+1}\sqb{k}=\frac{\beta_{K_1}^N\sqb{{k+1}}-\beta_{K_1}^N\sqb{{k}}}{t_{k+N+1} -  t_{k}},
\label{eq:mu_betak}
\end{equation}
for $N\geq2$, where $$\mu_{K_1}^2\sqb{k}=\frac{\ind_{\cb{K_1-1}}\rb{k}}{t_{k+2}-t_{k}} \mbox{ and } \beta_{K_1}^2\sqb{k}=\frac{\ind_{\cb{K_1}}\rb{k}}{t_{k+2} - t_{k}}.$$
\end{prop}
}

\begin{proof}
We compute $E_g\rb{t}$  {by integrating $\varepsilon_g\rb{t}$ defined in} \eqref{eq:zt}
\begin{equation}
E_g\rb{t}=\rb{2\lambda-h} s_1 \ind_{\left[\tau_1,\infty\right)}\rb{t}\cdot \rb{t-\tau_1}.
\end{equation}
We evaluate $\NFD^N E_g\rb{t_k}$ for $N=1,2$ as below
\begin{align*}
    \frac{\NFD^1 E_g\rb{t_k}}{s_1\rb{2\lambda-h}} &=  p\cdot  \ind_{\cb{K_1}}\rb{k}+\ind_{\cb{K_1+1,\dots,\infty}}\rb{k},\\
    \frac{\NFD^2 E_g\rb{t_k}}{s_1\rb{2\lambda-h}} &=  \frac{p}{t_{K_1+1} -  t_{K_1-1}}  \ind_{\cb{K_1-1}}\rb{k}+\frac{1-p}{t_{K_1+2} - t_{K_1}}\ind_{\cb{K_1}}\rb{k}  = \mu_{K_1}^{2}\sqb{k} p + \beta_{K_1}^{2}\sqb{k} \rb{1-p},
\end{align*}
and the result follows via \eqref{eq:NFD} and \eqref{eq:mu_betak}.
\end{proof}

 {Using \eqref{eq:mu_k_beta_k}}, the non-zero values of $\NFD^N E_g$ therefore satisfy 
\begin{equation}
\label{eq:supp_Eg}
\mathrm{supp} \rb{\NFD^N E_g\rb{t_k}} \subseteq  \mathbb{S}_N,    
\end{equation}
where $\mathbb{S}_N=\cb{{K_1-N+1},\dots,{K_1}}$.
We only have access to $\NFD^N X\rb{t_k}$, which also includes the effect of $\NFD^N G\rb{t_k}$. In the first recovery stage, we  compute the support of $\NFD^N E_g\rb{t_k}$ via thresholding. In the case of uniform sampling, a similar strategy was proposed with a fixed threshold $\th^N=\frac{\lambda_h}{2N}$  {\cite{Florescu:2021:J}}. In the present case of nonuniform sampling, we can achieve input reconstruction in a more general scenario by allowing a time-varying threshold $\th^N\sqb{k}$, which depends on the nonuniform trigger times $\cb{t_k}$. In order to constrain the effect of $\NFD^N G\rb{t_k}$ on the measurements $\NFD^N X\rb{t_k}$, the following lemma (proof in Section \ref{sect:Proofs}) derives a bound for $\NFD^N G\rb{t_k}$ and subsequently computes support of $\NFD^N E_g\rb{t_k}$.
\begin{lem}
\label{lem:km_kM}
Assume that 
\begin{equation}
    \label{eq:bound_G}
    \vb{\NFD^N G\rb{t_k}}<\th^N\sqb{k},
\end{equation}
where $\th^N\sqb{k}$ is defined as
\begin{equation*}
    \th^N\sqb{k}=\lambda_h\cdot \min\cb{\phi_{0,k+N-1},\phi_{0,k+N-2},\phi_{1,k-1},\phi_{1,k}}
\end{equation*}
and sequences $\phi_{i,l}, i=0,1,$  {$l \in \N$}, satisfy
\begin{align}
\label{eq:def_phi01}
    \begin{split}
        \phi_{0,l}&=
        \frac{\M{l}{N}{l-N+1}\cdot\B{l}{N}{l-N+2}}
        {\M{l}{N}{l-N+1}+\B{l}{N}{l-N+2}+\M{l}{N}{l-N+2}},\\
        \phi_{1,l}&=
        \frac{\M{l}{N}{l-1}\cdot\B{l}{N}{l}}
        {\M{l}{N}{l-1}+\B{l}{N}{l}+\B{l}{N}{l-1}},
    \end{split}
\end{align}
where $\mu_l^N\sqb{k}$ and $\beta_l^N\sqb{k}$ are given in \eqref{eq:mu_betak}. Furthermore, let $\KN=\cb{k_m,k_M}$ be such that
\begin{gather}
\label{eq:km_kM}
\begin{split}
k_m=\min \mathbb{M}_N, \ \  k_M=\max \mathbb{M}_N,\\
\mathbb{M}_N\triangleq\cb{k\in\Z\ \vert \  \vb{\NFD^N Z\rb{t_k}}\geq\th^N\sqb{k}}.
\end{split}
\end{gather}
Then $k_m\in\cb{K_1-N+1,K_1-N+2}$ and $k_M\in\cb{K_1-1,K_1}$.
\end{lem}

The next theorem is our main result, showing that $s_1\in\cb{\pm1}$ can be correctly identified and $\tau_1$ can be computed with an accuracy dependent on the sampling rate.

\statetheoremsolid{0.95\textwidth}{
\begin{theo}[Estimation of Folding Times]
\label{th:convergence_tau}
Let $g\in\PW{\Omega}$,
$$\cb{t_k}_{k\in\Z}=\mathrm{ASDM}_{\delta,b} \rb{\MOh g}$$ 
and assume $\vb{\NFD^N G\rb{t_k}}<\th^N\sqb{k}, \forall k \in \mathbb{Z}$. Let $\KN = \{k_m, k_M\}$ defined as in \eqref{eq:km_kM}, and let   $\widetilde{s}_1$ be the estimation of the folding sign given by
\begin{equation}
\label{eq:s_1}
    \widetilde{s}_1=-\mathrm{sign} \sqb{\NFD^N X \rb{t_{k_m}}}.
\end{equation}
Let $T_\mathsf{min}, T_\mathsf{max}>0$ be two constants satisfying $T_\mathsf{min} \leq  t_{k+1}-t_k \leq T_\mathsf{max}$. The folding time estimation $\widetilde{\tau}_1$ is defined as follows.
\begin{enumerate}
    \item If $k_M-k_m=N-1$, let $\widetilde{\tau}_1=\frac{t_{k_M}+t_{k_M+1}}{2}$ and thus
\end{enumerate} 
\begin{equation}
\label{eq:tau1_err1}
    \widetilde{s}_1=s_1, \quad\vb{\tau_1-\widetilde{\tau}_1}\leq T_\mathsf{max}/2.
\end{equation}
\begin{enumerate}[start=2]
    \item If $k_M-k_m=N-2$, let $\widetilde{\tau}_1=t_{k_M+1}$. Then
\begin{equation}
\label{eq:tau1_err2}
    \widetilde{s}_1=s_1, \quad\vb{\tau_1-\widetilde{\tau}_1}\leq T_\mathsf{max}.
\end{equation}
\end{enumerate}
\begin{enumerate}[start=3]
    \item If $k_M-k_m=N-3$, then $\widetilde{\tau}_1=\frac{t_{k_M+1}+t_{k_M+2}}{2}$ and \eqref{eq:tau1_err1} holds.
\end{enumerate}
\end{theo}
}

Given that $t_{k+1}-t_k\leq \frac{2\delta}{b-g_\infty}, \forall k \in \mathbb{Z}$, by choosing $T_\mathsf{max}=\frac{2\delta}{b-g_\infty}$, then $\tau_1$ can be recovered with any accuracy for a small enough $\delta$ given by (\ref{eq:tau1_err1},\ref{eq:tau1_err2}).
We note that in 
USF the folding times $\cb{\tau_r}$ are estimated as $\tau_r\approx K_r T$, leading to a similar error as here, bounded by the sampling period $T$.

\subsection{The Case of Multiple Folds}
\label{subsect:MultipleFolds}

In the general case there are $R$ folding times $\cb{\tau_r}_{r=1}^R$, and each has an associated set $\mathbb{S}_N^r=\cb{t_{K_r-N+1},\dots,t_{K_r}}$, such that the non-zero values of $\NFD^N E_g\rb{t_k}$ satisfy \eqref{eq:supp_Eg_multi}
\begin{equation}
\label{eq:supp_Eg_multi}
\mathrm{supp} \rb{\NFD^N E_g\rb{t_k}} \subseteq{\bigcup}_{r=1,\dots,R}\mathbb{S}_N^r.
\end{equation}

\begin{figure}[t]
\begin{center}
\includegraphics[width=0.75\textwidth]{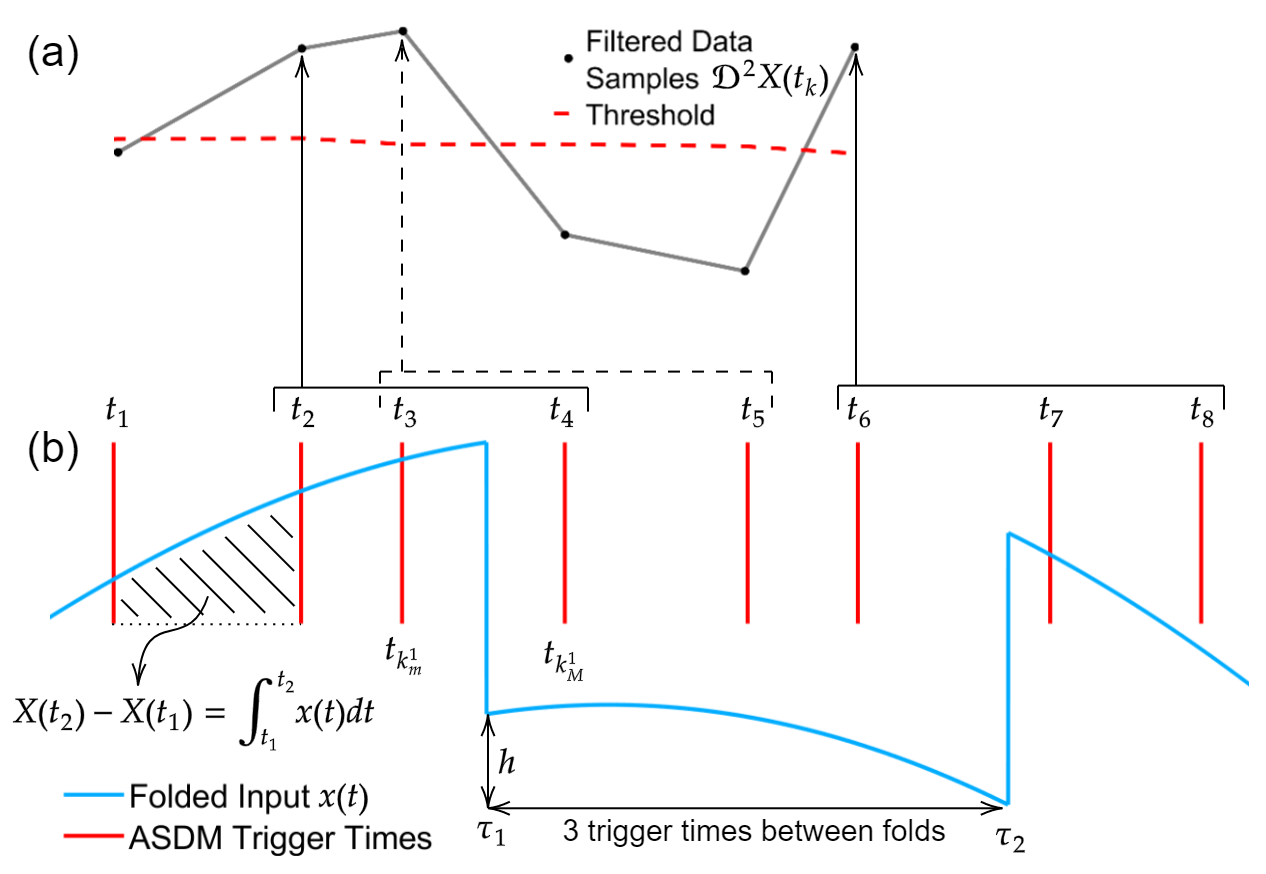} 
\caption{Illustration of the thresholding principle. (a) Each sample from the filtered data $\NFD^N X\rb{t_k}$ is computed from $N+1$ trigger times. For $N=2$, each folding time $\tau_r$ displaces up to two values of $\NFD^N X\rb{t_k}$ above threshold $\Psi\sqb{k}$. One can have a minimum of one value of $\NFD^N X\rb{t_k}$ above the threshold for each folding time. The minimum is achieved in this example because $\tau_2$ is close to $t_7$. (b) The separation of folding times $\tau_1$ and $\tau_2$ is enabled by the modulo-hysteresis parameter $h$. If at least $N+1=3$ trigger times are located between each two folds, the estimation can be performed sequentially to determine all folding times.}
\label{fig:thresholding_illustration}
\end{center}
\end{figure}

Here the results in Lemma \ref{lem:km_kM} and Theorem \ref{th:convergence_tau} can be applied sequentially provided that $\mathbb{S}_N^{r_1}\cap\mathbb{S}_N^{r_2}=\emptyset, \forall r_1,r_2=1\dots,R, r_1\neq r_2$. In other words, this means there are at least $N$ trigger times between each two consecutive folds $\tau_r, \tau_{r+1}$, which is guaranteed if $\tau_{r+1}-\tau_r\geq N T_\mathsf{max}$.
A sufficient condition can be obtained via Lemma \ref{lem:separation} 
\begin{equation}
    \label{eq:bound_separation}
    \frac{h^*}{\Omega g_\infty}\geq N T_\mathsf{max}.
\end{equation}
Then if (\ref{eq:bound_G},\ref{eq:bound_separation}) are true then $\cb{\tau_r}_{r=1}^R$ and $\cb{s_r}_{r=1}^R$ can be computed as follows. Let $\KN=\bigcup\nolimits_{r \in \cb{1,R}} {\cb{k_m^r, k_M^r}}$, where $k_m^r, k_M^r$ are computed as below. For $r=1$,
\begin{align}
\label{eq:k_m1_k_M1}
\begin{split}
k_m^1&=\min\cb{k\ \vert \ \vb{X_k^N} \geq \th^N\sqb{k}},\\
k_M^1&=\max\cb{k\leq k_m^1+\const{N-1} \ \vert \ \vb{X_k^N} \geq \th^N\sqb{k}}.
\end{split}
\end{align}
For $r=2,\dots,R$
\begin{align}
\label{eq:k_mr_k_Mr}
\begin{split}
k_m^r&=\min\cb{k> k_M^{r-1} \ \vert \ \vb{X_k^N} \geq \th^N\sqb{k}},\\
k_M^r&=\max\cb{k\leq k_m^r+\const{N-1} \ \vert \ \vb{X_k^N} \geq \th^N\sqb{k}},
\end{split}
\end{align}
where $X_k^N\triangleq\NFD^N X\rb{t_k}$. Then $\widetilde{s}_r, \widetilde{\tau}_r, r=1,\dots,R$ can be computed sequentially with Theorem \ref{th:convergence_tau} from $\KN$.
 {To better explain the proposed technique, we illustrated the thresholding strategy for multiple folds in \fig{fig:thresholding_illustration}. The diagram shows the role of $h$ in separating the folds which, in turn, enables an independent detection of the folding times. }
The following lemma provides sufficient recovery conditions.

\statelemsolid{0.95\textwidth}{
\begin{lem}[Sufficient Recovery Conditions] The recovery  {is possible when either of the conditions 1) and 2) below is true.}
\begin{enumerate}
    \item The values $\widetilde{s}_r, \widetilde{\tau}_r, r=1,\dots,R$ can be recovered from $\NFD^N X\rb{t_k}$ with an error given by (\ref{eq:tau1_err1},\ref{eq:tau1_err2}) if
\begin{enumerate}[leftmargin=1.5cm,label=$\mathrm{\cond}_{\arabic*}$)]
\item $\rb{\frac{C^2}{b-\lambda}2\delta\Omega e}^{N-1} g_\infty < \frac{\lambda_h}{CNe}$.
\item $N\frac{2\delta \Omega g_\infty}{b-\lambda}\leq h^*$,
\end{enumerate}
where $h^\ast=\min\cb{h,2\lambda-h}$ and  $C=\frac{b+\lambda}{b-\lambda}$.
\item Furthermore, $\mathrm{\cond}_1)$ and $\mathrm{\cond}_2)$ are satisfied if
\begin{equation}
    \delta<\frac{\rb{b-\lambda}h^*\kappa}{2N\Omega g_\infty} \mbox{ where } 
    \kappa  = \min \left\{ {1,\frac{{{\lambda _h}}}{e^2 h^*C^{2+\frac{1}{N-1}}}} \right\}.
\end{equation}
\end{enumerate}
\label{lem:sufficient_conditions}
\end{lem}
}

\begin{proof}
\begin{enumerate}
    \item We use the upper bound for $\vb{\NFD^N G\rb{t_k}}$ in  {Proposition} \ref{prop:boundDG} to derive a sufficient condition for \eqref{eq:bound_G} as
\begin{equation}
\label{eq:th1}
\rb{\frac{T_\mathsf{max}^3}{T_\mathsf{min}^2}\Omega e}^N g_\infty <\frac{\lambda_h}{N e}\frac{T_\mathsf{min}}{T_\mathsf{max}}.    \end{equation}
 {To prove \eqref{eq:th1}}, note that, in the limit case, $T_\mathsf{max}=T_\mathsf{min}=T$ is always satisfied for large enough $N$ if $$\frac{T_\mathsf{max}^3}{T_\mathsf{min}^2}\Omega e=T \Omega e<1.$$
The latter is the sampling rate condition for uniform sampling \cite{Bhandari:2017}. Then $\left.{\mathrm{\cond}}_1\right)$ follows from \eqref{eq:th1} and $\left.{\mathrm{\cond}}_2\right)$ from \eqref{eq:bound_separation} where $T_\mathsf{min}$ and $T_\mathsf{max}$ are selected as the following bounds for the intervals between the trigger times \cite{Lazar:2004} \begin{equation}
T_\mathsf{min}=\frac{2\delta}{b+\lambda}  \leq \vb{t_{k+1}-t_k} \leq \frac{2\delta}{b-\lambda}=T_\mathsf{max}.
\label{eq:ISIbounds}
\end{equation}

\item From $\left.{\mathrm{\cond}}_1\right)$ and $\left.{\mathrm{\cond}}_2\right)$ it follows that $2\delta<\min\cb{B_1\rb{N}, B_2\rb{N}}$, where 
\begin{equation*}
  B_1\left( N \right) = \frac{b - \lambda }{C^2\Omega e}\left( \frac{\lambda _h}{g_\infty CNe} \right)^\frac{1}{N - 1}\ \text{and}\ B_2\left( N \right) = \frac{\left( {b - \lambda } \right)h^*}{N\Omega g_\infty }.
\end{equation*}
A direct calculation yields $B_1\rb{N} \frac{e^2 h^* C^2\sqrt[N-1]{C}}{\lambda_h }\geq B_2\rb{N}$, which completes the proof.
\end{enumerate}
\end{proof}

 {If $g_\infty$ is unknown yet an upper bound is known},  {the results in Lemma \ref{lem:sufficient_conditions} still hold true by replacing $g_\infty$ with the upper bound. Furthermore, we note that the conditions in Lemma \ref{lem:sufficient_conditions} are only sufficient, and can be conservative. Therefore, in practice, one can achieve good reconstructions with higher values of $N$, some of which that might not satisfy the lemma. }

\subsection{Input Reconstruction}
\label{subsect:InputRec}

Assuming $\cb{\widetilde{\tau}_r,\widetilde{s}_r}$ known, the estimation of $E_g\rb{t_k}$ is
\begin{equation}
\label{eq:Eps_tilde}
    \widetilde{E}_g\rb{t_k}=2\lambda_h \sum\limits_{r=1}^R \widetilde{s}_r \ind_{\left[\widetilde{\tau}_r,\infty\right)}\rb{t_k}\cdot \rb{t_k-\widetilde{\tau}_r}.
\end{equation}
Finally, $\widetilde{g}\rb{t}$ is recovered recursively using reconstruction from local averages,
 $$\Delta \widetilde{G}\rb{t_k}=\widetilde{G}\rb{t_{k+1}}-\widetilde{G}\rb{t_{k}}\EQc{eq:rec_local_av}\int_{t_k}^{t_{k+1}}g\rb{s} ds,$$ 
 where $\widetilde{G}\rb{t_{k}}=X\rb{t_{k}}+\widetilde{E}_g\rb{t_{k}}$. This recovery approach was proposed in the context of irregular sampling in \cite{Feichtinger:1994} and adapted for event-driven ASDM sampling in \cite{Lazar:2004}.  In our case, sequence $\widetilde{g}_n\rb{t}$ is computed as
\begin{align}
\begin{split}
    \widetilde{g}_0\rb{t}&=\syntharg{\Delta \widetilde{G}\rb{t_k}}\rb{t},\\
    \widetilde{g}_{n+1}\rb{t}&=\widetilde{g}_{n}\rb{t}+\widetilde{g}_{0}\rb{t}-\syntharg{\locav \widetilde{g}_{n}}\rb{t},
\end{split}    
\label{eq:local_av_rec}
\end{align}
where,
\begin{itemize}
  \item  $\synth:\ell^2\rightarrow L^2\rb{\mathbb{R}}$.
  \item $\syntharg{a_k}\triangleq\sum\limits_{k\in\mathbb{Z}} a_k \cdot \mathrm{sinc}_\Omega \rb{t-s_k}$, $s_k=\frac{t_k+t_{k+1}}{2}$, and, 
  \item $\locav:L^2\rb{\mathbb{R}}\rightarrow \ell^2, \rb{\locav f}_k=\int\limits_{t_k}^{t_{k+1}} f\rb{s} ds$.
\end{itemize}
We outline the recovery procedure in Algorithm \ref{alg:1}.

\begin{algorithm}[!t]
\SetAlgoLined
{\bf Data:} $\lambda,h,\delta,b\in \mathbb{R}_+, \{t_k\}, \Omega>0, N\geq2$.\\
\KwResult{ $\widetilde{g}\rb{t}$ }
 \begin{enumerate}
 \itemsep 2pt
\item Compute $X_k^N=\NFD^N \rb{\sum\limits_{m= {0}}^{k-1} q_m} $, where $q_m = \rb{-1}^m 2\delta-\rb{-1}^m b\rb{t_{m+1}-t_m}$.
\item Compute $R,\ \cb{k_m^r,k_M^r}_{r=1}^R$ with \eqref{eq:k_m1_k_M1}, \eqref{eq:k_mr_k_Mr}.
\item For $r=1,\dots,R$, compute $\widetilde{s}_r$ with \eqref{eq:s_1} and,
$$\widetilde{\tau}_r = \frac{1}{2}
\begin{cases}
t_{k_M^r}+t_{k_M^r+1} & \mbox{ if } \rb{k_M^r-k_m^r}=N-1 \\ 
2\rb{t_{k_M^r+1}} &  \mbox{ if } \rb{k_M^r-k_m^r}=N-2 \\ 
t_{k_M^r+1}+t_{k_M^r+2} &  \mbox{ if } \rb{k_M^r-k_m^r}=N-3
\end{cases}.$$
\item Compute $\widetilde{E}_g\rb{t_k}$ using \eqref{eq:Eps_tilde} and $\widetilde{G}\rb{t_k}$ using \eqref{eq:expanded_Z}.
\item Compute $\widetilde{g}_n\rb{t}$ from samples $ \Delta \widetilde{G}\rb{t_{k}}$ with \eqref{eq:local_av_rec}.
\end{enumerate}
\caption{Recovery Algorithm.}
\label{alg:1}
\end{algorithm}
The following result proven in Section \ref{sect:Proofs} gives a bound for the input reconstruction error.

\statepropsolid{0.95\textwidth}{
\begin{prop}
\label{prop:final_error_bound}
Assuming that the conditions in Lemma \ref{lem:sufficient_conditions} are satisfied and  {that} the data consists of $R$ folding times, then the input can be recovered recursively as $\widetilde{g}_n\rb{t}$, where $n$ is the iteration number. The reconstruction error satisfies:
\begin{equation}
\label{eq:final_error_bound}
    \norm{g-\widetilde{g}_n}_{L^2} \leq  \frac{4\lambda_h 2\delta R \sqrt{\Omega\pi}}{\pi\rb{b-\lambda}-2\delta\Omega}+\rb{\frac{2\delta\Omega}{\pi\rb{b-\lambda}}}^{n+1}\norm{g}_{L^2}.
\end{equation}
\end{prop}
}

We note that the second term in the error bound \eqref{eq:final_error_bound} is guaranteed to vanish for $n\rightarrow\infty$ if $\frac{2\delta \Omega}{\pi\rb{b-\lambda}}<1$, which is satisfied if condition $\mathrm{\cond}_2)$  in Lemma \ref{lem:sufficient_conditions} is true. Furthermore, the error in \eqref{eq:final_error_bound} can be made arbitrarily small for a small enough ASDM threshold $\delta$.  {Furthermore, we note that Algorithm \ref{alg:1} has an intrinsic method to indicate that the data does not verify the recovery conditions, if $k_M^r-k_m^r\not\in\cb{N-3,N-2,N-1}$ for any $r\in\cb{1,\dots,R}$. Although theoretically there could be situations where the algorithm runs with no error for an input not satisfying the recovery conditions, we  found this to be very unlikely in practice. }

\section{Numerical and Hardware Experiments}
\label{sect:NumericalDemo}

We  {validate} the proposed event-driven model using both synthetic and real data measurements. We generate a bandlimited function $g\in\PW{\Omega}$, which is encoded directly using the standalone ASDM into trigger times $t_k^\mathsf{ASDM}$. The same input is then processed with the proposed encoding model, which triggers times $t_k^\mathsf{MEDS}$. The direct reconstruction from ASDM samples is denoted $\widetilde{g}_\mathsf{ASDM}\rb{t}$, and the proposed reconstruction from the MEDS samples is denoted $\widetilde{g}_\mathsf{MEDS}\rb{t}$.
We measure the relative recovery error via the ratio $\mathsf{Err}\rb{\widetilde{g},g}$, defined as
\begin{equation}
	\mathsf{Err}\rb{\widetilde{g},g}=100\cdot\frac{\norm{g-\widetilde{g}}_{L^2}}{\norm{g}_{L^2}} \rb{\%}.
\end{equation}
We denote the recovery errors with each method by $\mathsf{Err}_{\mathsf{ASDM}}=\mathsf{Err}\rb{\widetilde{g}_\mathsf{ASDM},g}$ and $\mathsf{Err}_{\mathsf{MEDS}}=\mathsf{Err}\rb{\widetilde{g}_\mathsf{MEDS},g}$, respectively.

The recovery results from synthetic data are in \ref{sect:experiments_synthetic}, and the ones from experimental data are presented in \ref{sect:experiments_measured}.

\subsection{Reconstruction using Synthetic Data}
\label{sect:experiments_synthetic}

Here we compare the proposed model with the ASDM encoder. The input is $g\rb{t}=Ag_0\rb{t}$, where
\begin{equation*}
    g_0\rb{t}=\sum\nolimits_{n\in\mathbb{Z}} c_n \mathrm{sinc}_{\Omega} \rb{ t-n\pi/\Omega},
\end{equation*}
where $\Omega=150\ \mathrm{rad/s}$ and $c_n$ are drawn from the uniform distribution $U\rb{\sqb{-1,1}}$.
The signal is truncated such that $t\in\sqb{0,0.13\ \mathrm{s}}$, and normalised such that $\mathrm{max}_t\vb{g_0\rb{t}}=1$. The ASDM parameters are $\delta = 2.5\times 10^{-3}, b=9$. 
The Nyquist rate condition is $\vb{g\rb{t}}<g_\textsf{MAX}$, where  {$g_\textsf{MAX}$  {denotes} the ASDM dynamic range} $g_\textsf{MAX}=\frac{\pi b-2\delta \Omega}{\pi}=8.76$.

The proposed encoder consists of a modulo model with threshold $\lambda = g_\textsf{MAX}/2$ and hysteresis parameter $h=\lambda/2$ in series with an ASDM whose parameters are the same as the standalone ASDM encoder. Therefore the models are equivalent for $\vb{g\rb{t}}<\lambda$. For $A=34.6$, the ASDM generates trigger times $\cb{t_k^\textsf{ASDM}}_{k=1}^{14}$ in response to input $g\rb{t}$ when its amplitude is within the ASDM dynamic range $\sqb{-8.76,8.76}$, and sends no information when the input amplitude is outside the range. However, the proposed model generates trigger times $\cb{t_k^\textsf{MEDS}}_{k=1}^{220}$, which cover the whole duration of $g\rb{t}$, as depicted in \fig{fig:synth_data} (a). The input is first recovered from the ASDM samples $\cb{t_k^\textsf{ASDM}}_{k=1}^{14}$  \cite{Lazar:2004}, which leads to a high reconstruction error especially where the input has a large amplitude, as depicted in \fig{fig:synth_data}(d). For the proposed model, we used the $\NFD^3$ to identify correctly the $21$ folding times. The filtered data $\NFD^3 X\rb{t_k^\textsf{MEDS}}$, the threshold $\Psi_N\sqb{k}$ and the estimated folding times are  in \fig{fig:synth_data} (b,c). We then recovered $\widetilde{g}_{\mathsf{MEDS}}\rb{t}$ using Algorithm \ref{alg:1}. The resulted errors for the two methods are $\mathsf{Err}_{\textsf{ASDM}}=6.11\cdot10^3\%$ and $\mathsf{Err}_{\textsf{MEDS}}=0.25\%$.

 {To see the effect of changing the ASDM threshold $\delta$ on the recovery performance, we repeated the recovery for $10$ uniformly spaced values of $\delta$ in interval $\sqb{10^{-3},3\cdot 10^{-3}}$. We evaluated the input error as before, the number of spikes generated, and the folding time error evaluated as $\mathsf{Err}_{\tau}=100\cdot\tfrac{\norm{\tau-\widetilde{\tau}}_{\ell^2}}{\norm{\tau}_{\ell^2}} \rb{\%}$. The results, depicted in \fig{fig:varying_Delta}, show that when $\delta>2.6\cdot 10^{-3}$ the recovery algorithm does not identify the folds correctly leading to a large error. Choosing $\delta<1.8$ leads to an exponential increase in sample size without a significant decrease in error. Moreover, we remark that, in this example, the input error is around $6$ times larger than the folding time error. This is in line with the theoretical derivation that the error for computing $\Delta G\rb{t_k}$ is $2\lambda_h=6.57$ times larger than the folding time error as shown in \eqref{eq:Gtk_err}.

 \begin{figure*}[!t]
\centering
\begin{tabular}{cc}
\includegraphics[width=0.5\textwidth]{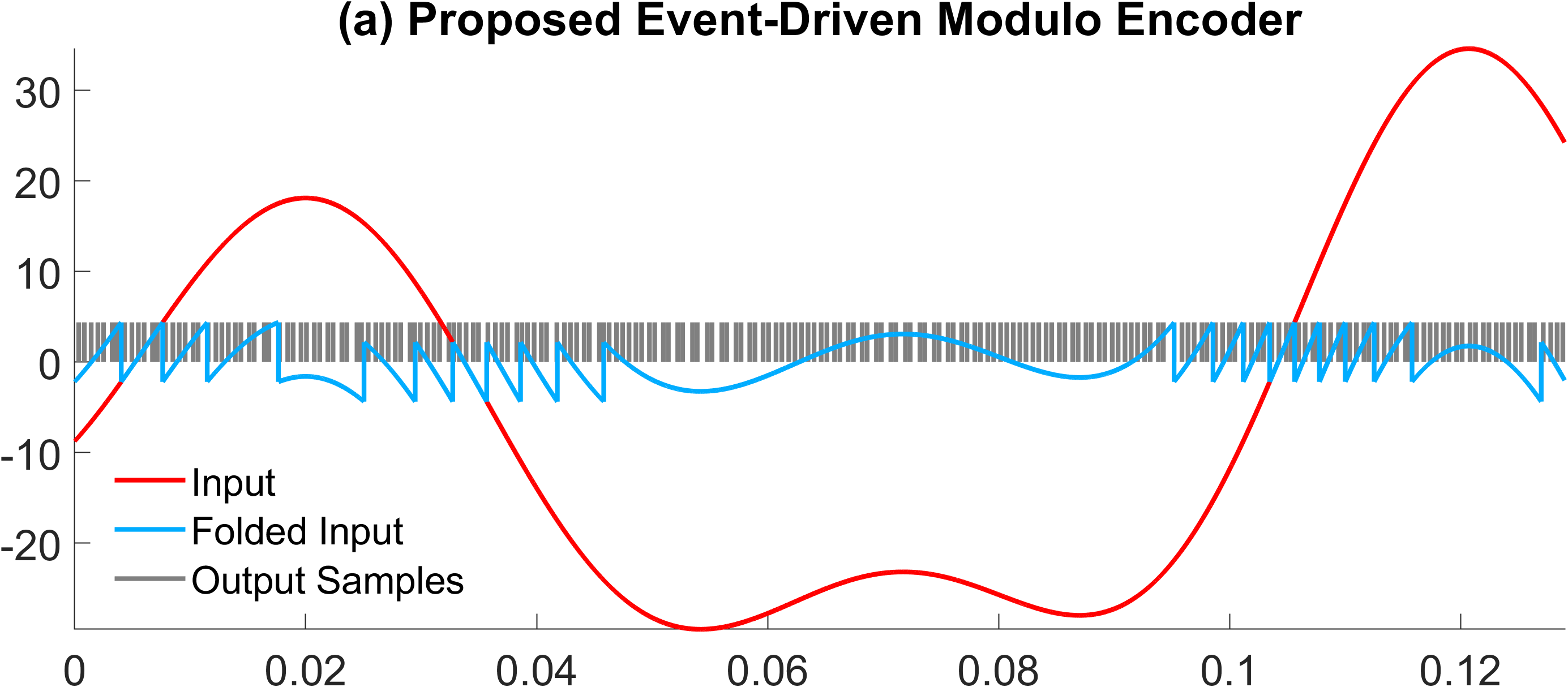} & \includegraphics[width=0.5\textwidth]{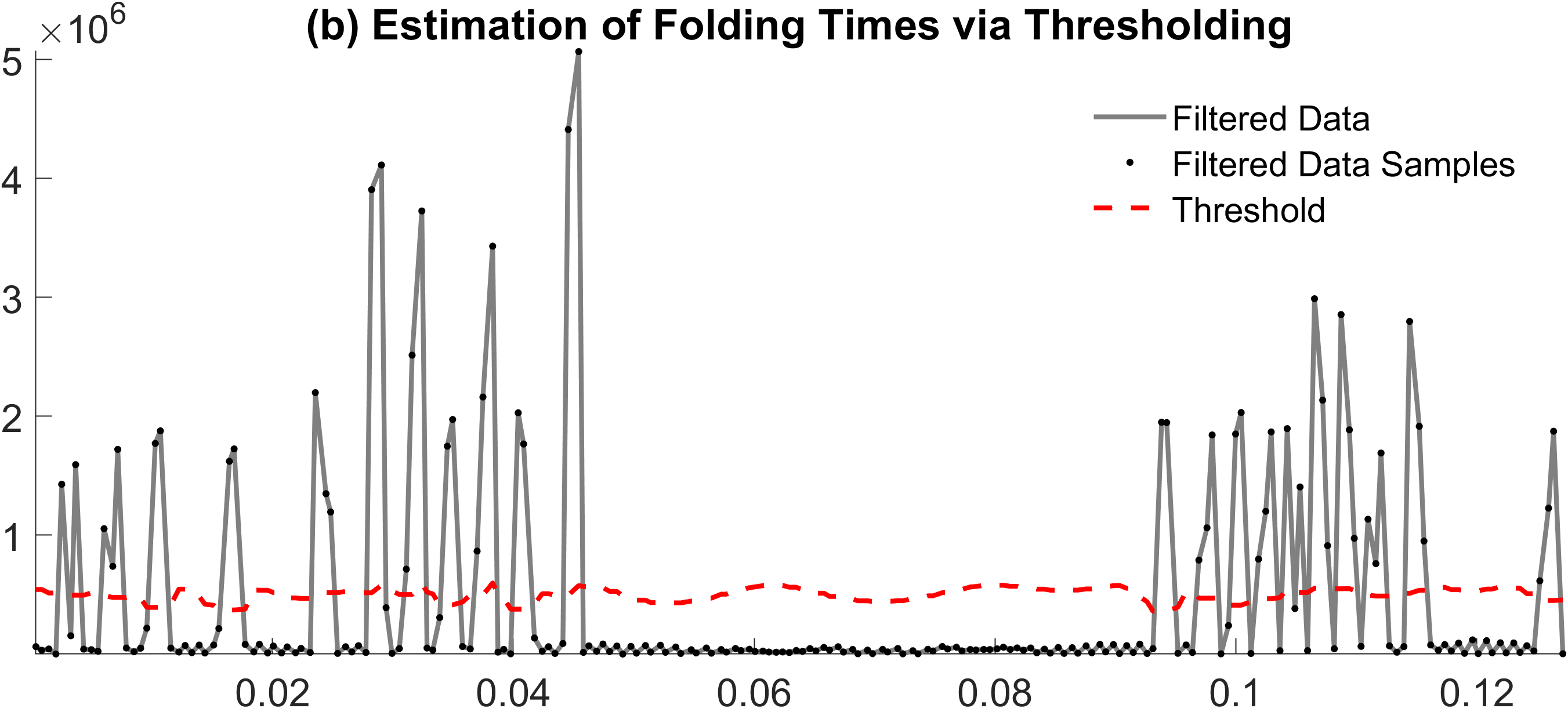} \\
\includegraphics[width=0.5\textwidth]{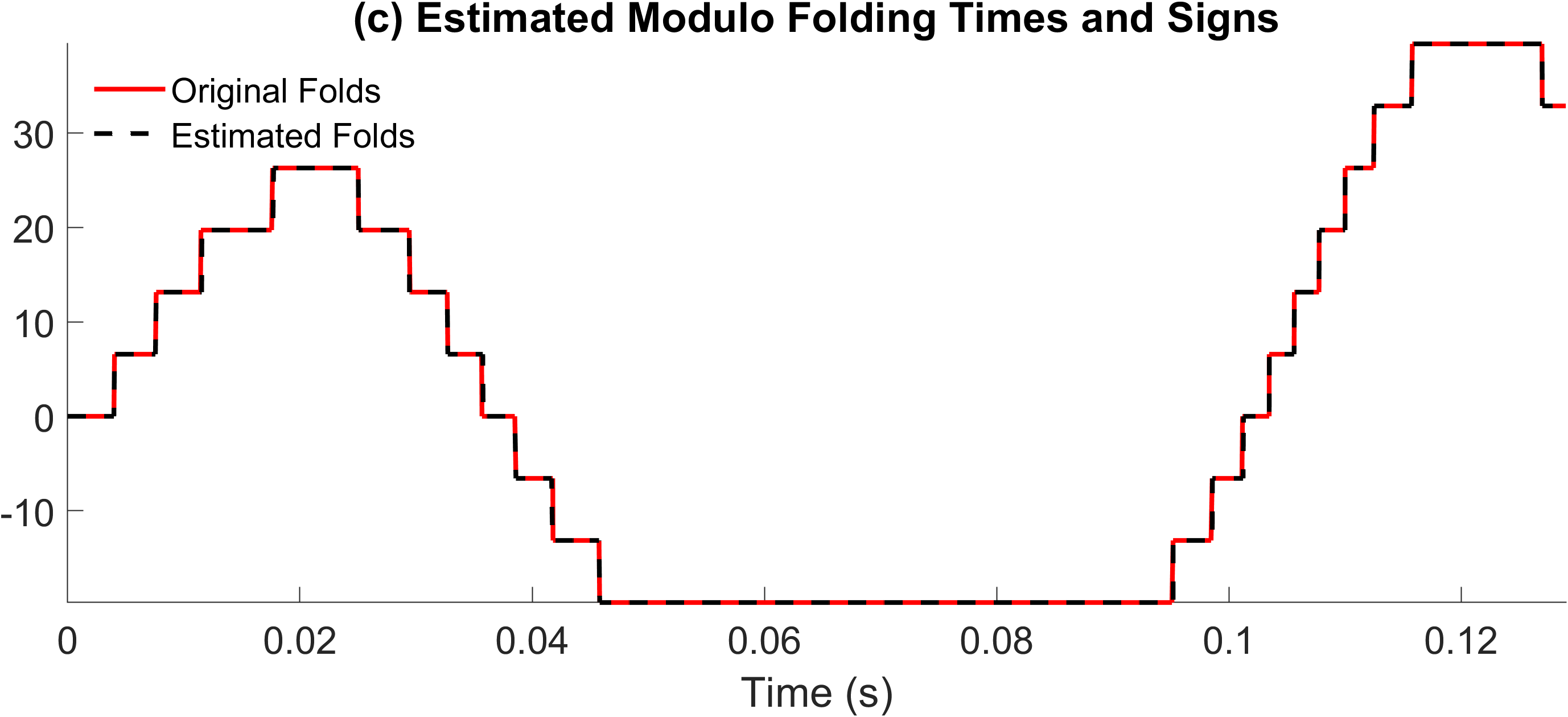} & \includegraphics[width=0.5\textwidth]{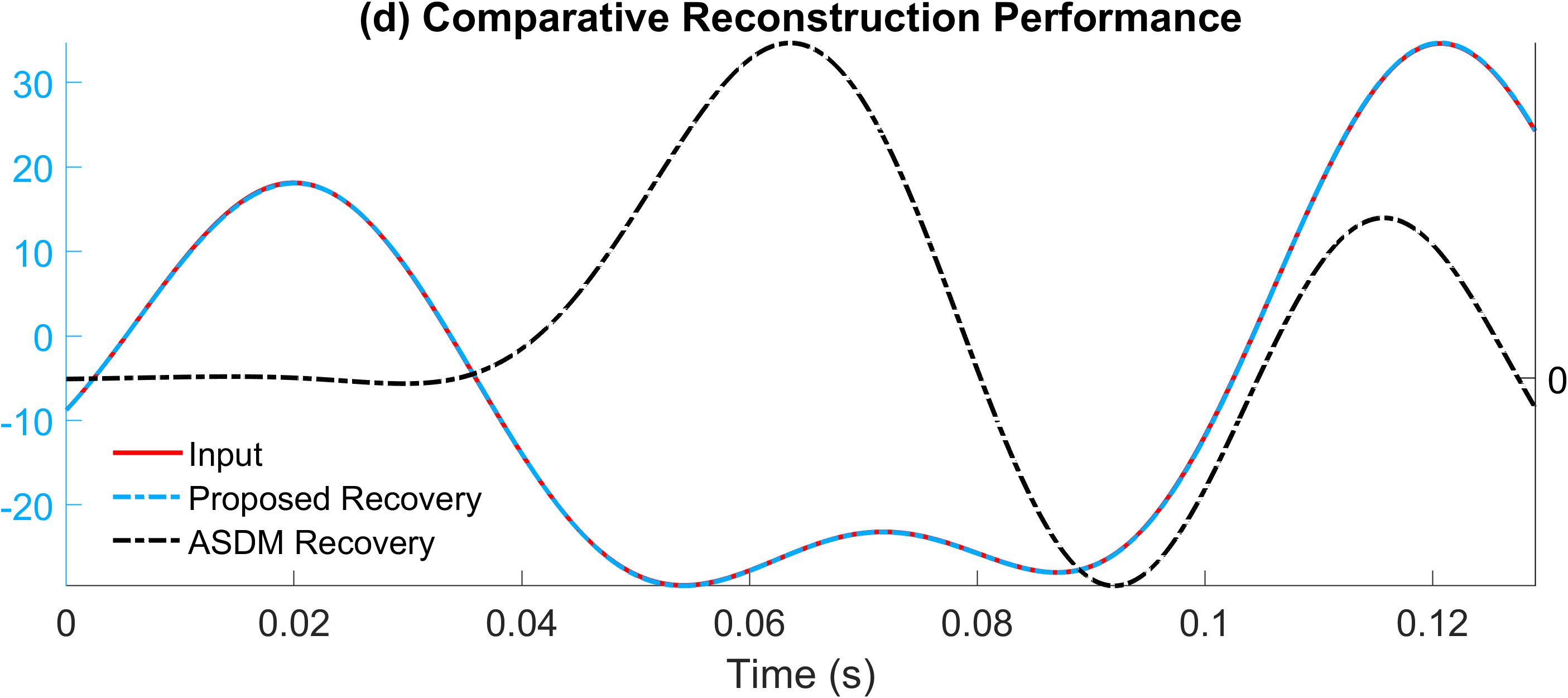}
\end{tabular}
\caption{Simulations with synthetic data. (a) Input, output of the modulo nonlinearity and corresponding output trigger times. (b) Estimating the folding times by thresholding the filtered data. (c) Reconstruction of the residual function. (d) Input, proposed reconstruction and conventional IF based recovery.}	
\label{fig:synth_data}
\end{figure*}

\begin{figure}[!t]
\begin{center}
    \includegraphics[width=.30\textwidth]{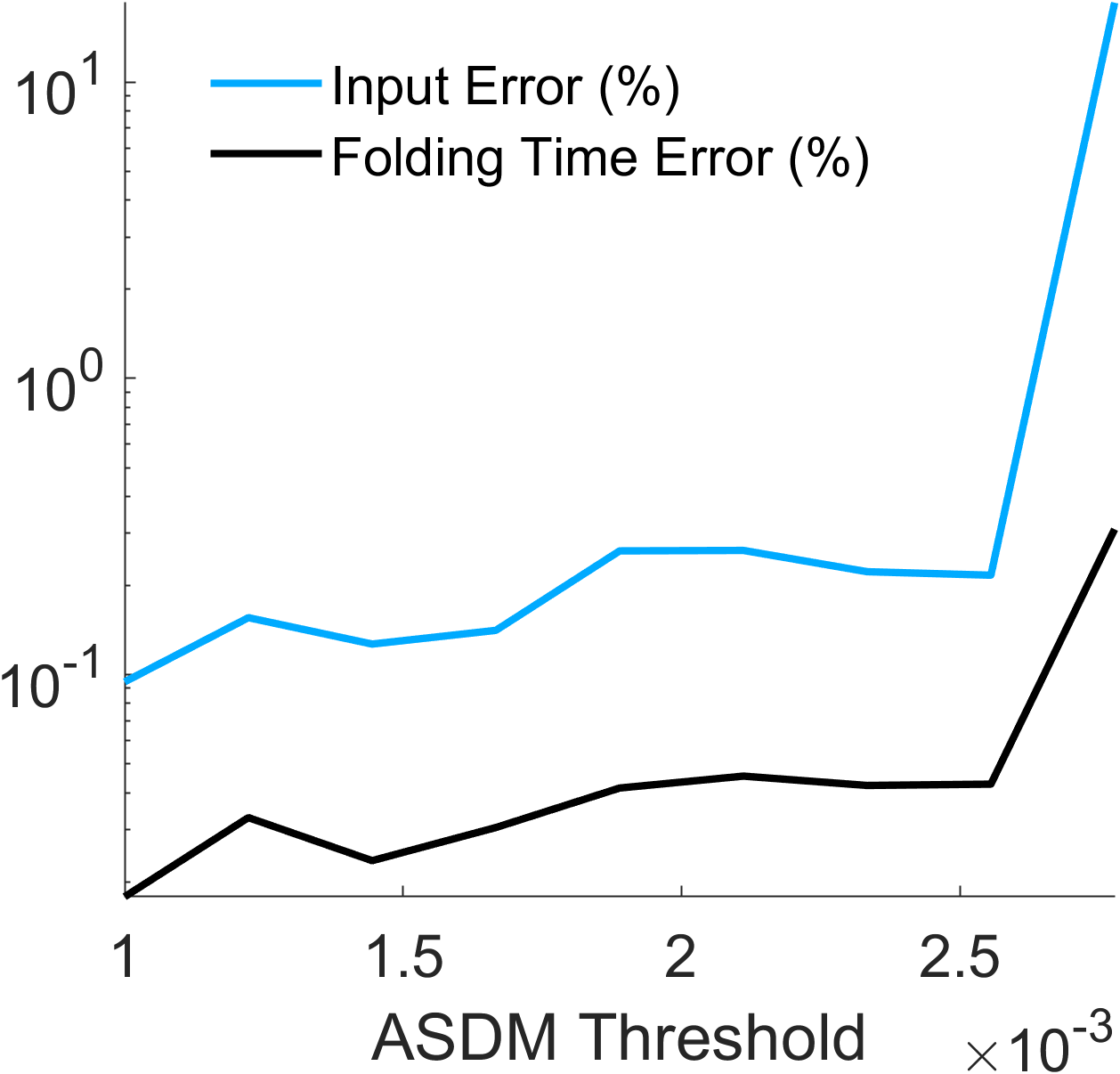}
    \includegraphics[width=.30\textwidth]{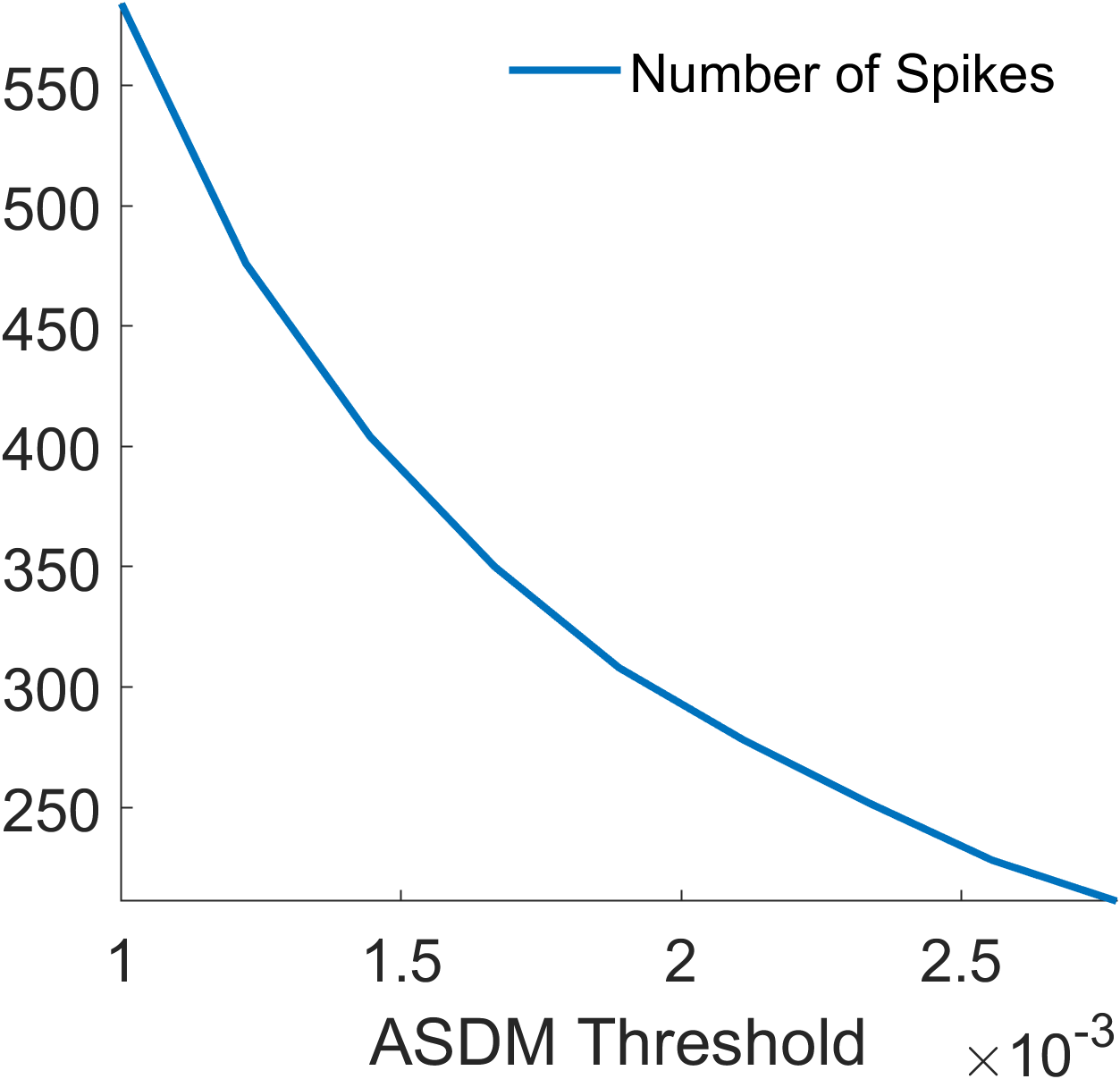}
\end{center}
\caption{The effect of the ASDM threshold on recovery accuracy and number of samples.}	
\label{fig:varying_Delta}
\end{figure}
}

\subsection{Reconstruction using Experimental Data}
\label{sect:experiments_measured}

In this example, the input is $g\rb{t}=4.51\cdot\sin\rb{\Omega \rb{t-\tau_0}}$ where $\Omega=125\ \mathrm{rad/s}$ and $\tau_0=1.6\cdot10^{-2}$. The output data was generated using a MEDS hardware prototype shown in \fig{fig:hardware} that is coarsely calibrated such that the modulo output is within the ASDM dynamic range. The corresponding acquisition pipeline is shown in \fig{fig:diagram_hardware}. The MEDS parameters, estimated using line search based optimization \cite{Florescu:2022:Ca}, are $\lambda = 1.53$, $h=1.51$, $\delta = 2.07\times 10^{-4}$, and $b=2.22$. Therefore the ASDM dynamic range is $g_\mathsf{MAX}=\frac{\pi b - 2\delta \Omega}{\pi}=2.206$.
The output of the prototype in response to $g\rb{t}$ is  $\cb{t_k^\textsf{MEDS}}_{k=1}^{214}$ depicted in \fig{fig:real_data} (a). 

\begin{figure}[!t]
\begin{center}
\includegraphics[width=.65\textwidth]{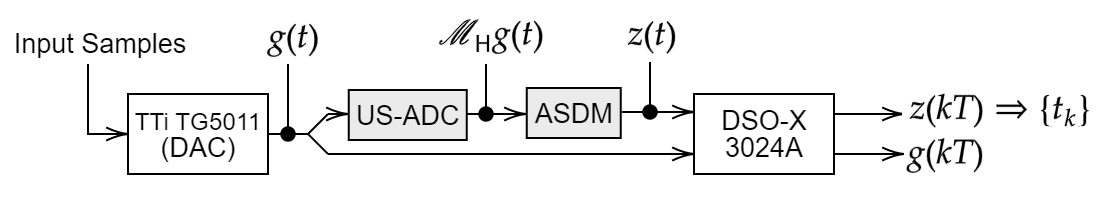}
\caption{The diagram of the hardware setup used for experiments. The input $g\rb{t}$ is generated via a digital-to-analog converter (DAC) from digital samples. The MEDS hardware prototype is implemented using an US-ADC in series with an ASDM circuit. The continuous MEDS output $z\rb{t}$ and input $g\rb{t}$ are then sampled with an oscilloscope to generate uniform samples $z\rb{kT}$ and $g\rb{kT}$. The ASDM trigger times $\cb{t_k}$ are subsequently computed from the zero-crossings of $z\rb{kT}$.}	
\label{fig:diagram_hardware}
\end{center}
\end{figure}

\begin{figure}[!t]
\centering
\includegraphics[width=0.75\textwidth]{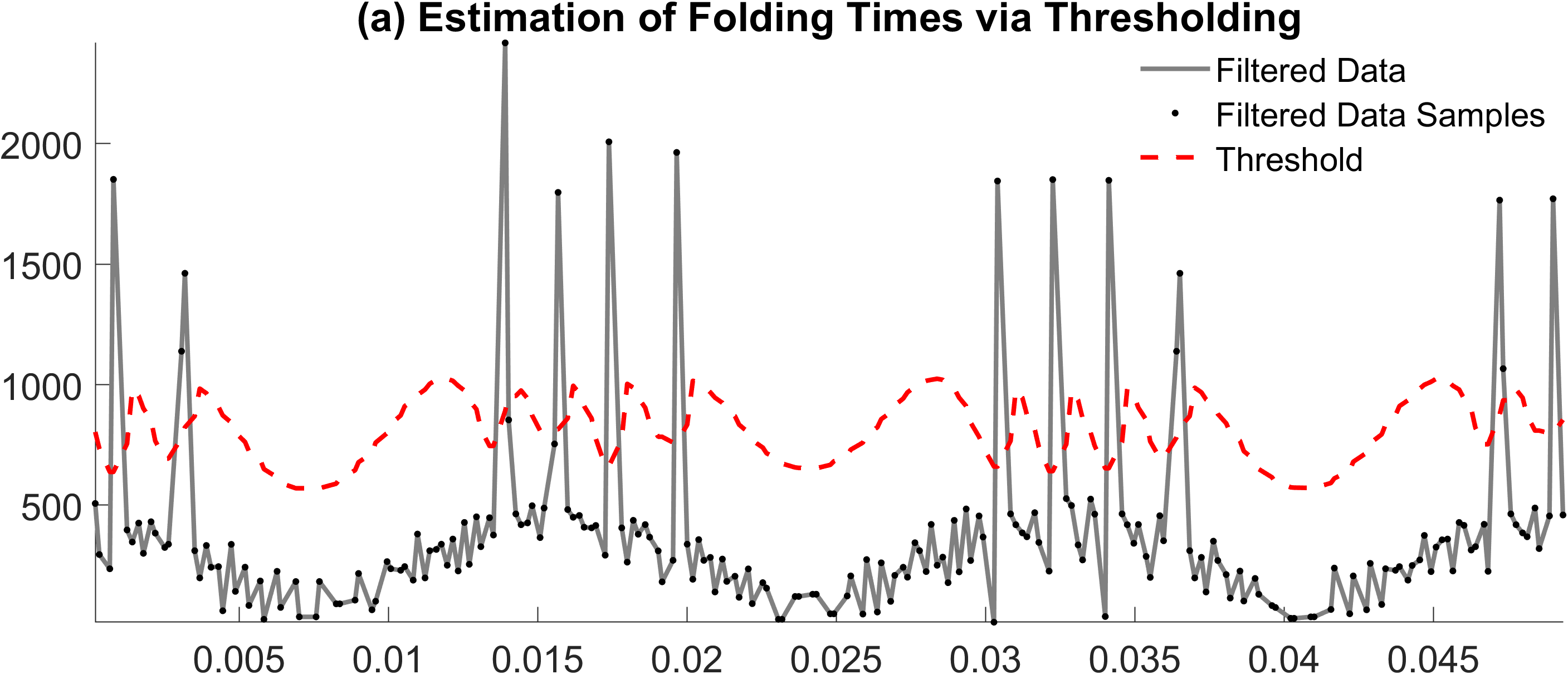}\vspace{0.35cm}\\
\includegraphics[width=0.7\textwidth]{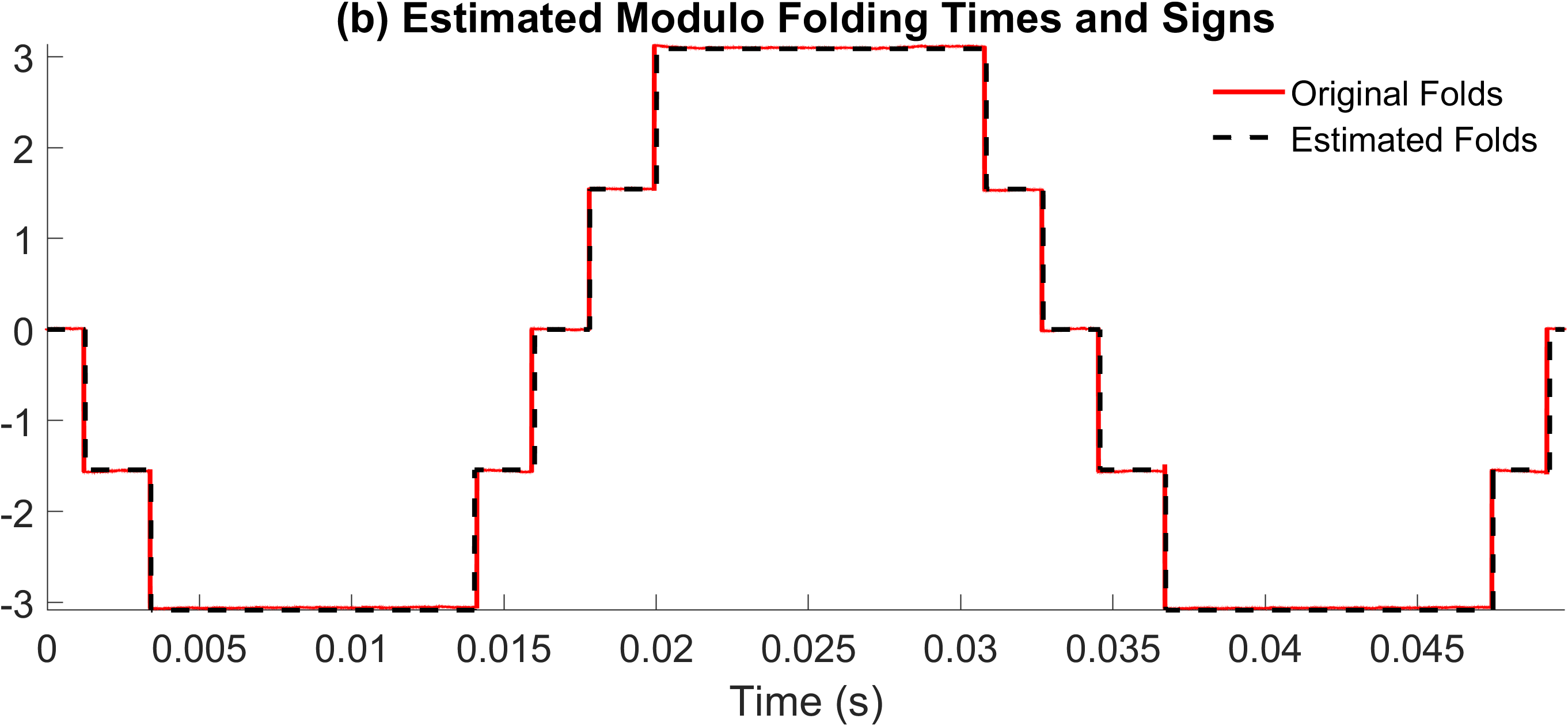}\\
\includegraphics[width=0.7\textwidth]{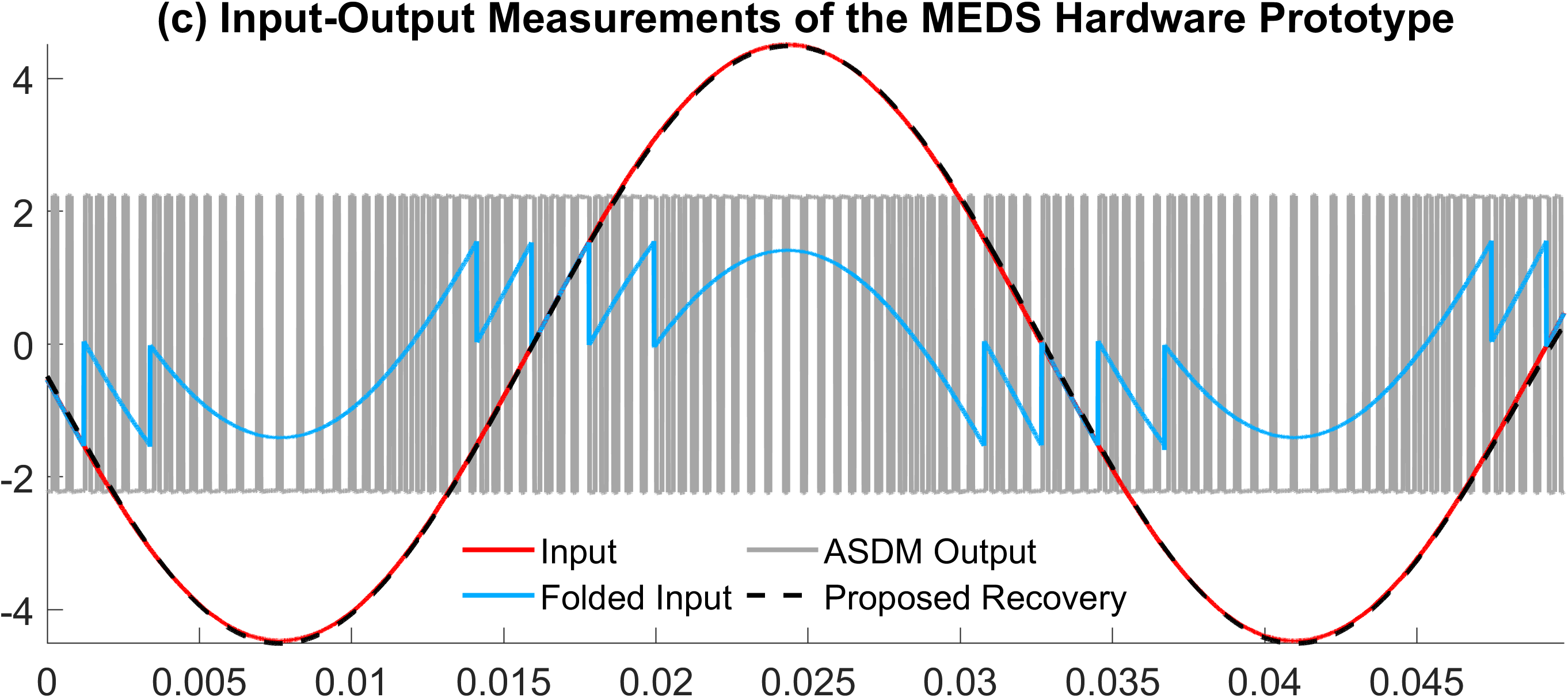}\vspace{0.35cm}
\caption{Input reconstruction from measurements based on our MEDS hardware (\fig{fig:hardware}). (a) Estimating the folding times by thresholding the filtered data. (b) Reconstruction of the residual function. (c) The proposed reconstruction superimposed on the input data.}	
\label{fig:real_data}
\end{figure}

We used the $\NFD^2$ to identify correctly the $12$ folding times. The filtered data $\NFD^2 X\rb{t_k^\textsf{MEDS}}$, the threshold $\Psi_N\sqb{k}$ and the estimated folding times are  in \fig{fig:real_data} (b,c). The input is recovered with Algorithm \ref{alg:1}. 
The resulted error is $\mathsf{Err}_{\textsf{MEDS}}=0.68\%$.

\section{Proofs}
\label{sect:Proofs}
Here we prove Proposition \ref{prop:boundDG}, Lemma \ref{lem:km_kM}, Theorem \ref{th:convergence_tau}, and Proposition \ref{prop:final_error_bound}.

\begin{proof}[Proof of Proposition \ref{prop:boundDG}]
We fix an arbitrary value $k\in\Z$ and expand $G(t_{k+l}),\quad l\in$  {$\cb{0,\dots,N}$} in Taylor series around $\tau=\frac{t_{k+N}+t_k}{2}$
\begin{align}
    G\rb{t_{k+l}}=G(\tau)+&\sum\limits_{n=1}^{N-1} \rb{t_{k+l}-\tau}^n g^{\rb{n-1}}\rb{\tau}+\frac{\rb{t_{k+l}-\tau}^N}{N!}g^{\rb{N-1}}\rb{\zeta_{k+l}},    
\end{align}
where $\zeta_{k+l}$ is between $t_{k+l}$ and $\tau$. The equation above uses that $G^{\rb{1}}(t)=g(t).$
Let $r\rb{k+l}$ denote the Taylor series remainder $r\rb{k+l}=\frac{\rb{t_{k+l}-\tau}^N}{N!}g^{\rb{N-1}}\rb{\zeta_{k+l}}$. We use the fact that  {the} operator $\NFD^N$ cancels out polynomials of degree up to $N-1$, and thus
\begin{align*}
 \vb{\NFD^N G\rb{t_k}}=\vb{\NFD^N r\rb{k}}&=\vb{\frac{\NFD^{N-1} r\rb{k+1}-\NFD^{N-1}r\rb{k}}{t_{k+N}-t_k}}\\ 
 &\leq\frac{\vb{\NFD^{N-1}r\rb{k+1}}+\vb{\NFD^{N-1}r\rb{k}}}{t_{k+N}-t_k}\\
 & \leq \frac{2 \cdot \underset{l=0,1}{\max}\vb{\NFD^{N-1}r\rb{k+l}}}{N T_\mathsf{min}}.
\end{align*}

Similarly, via induction, it can be shown that
\begingroup
\allowdisplaybreaks
\begin{align*}
\begin{split}
    \vb{\NFD^N G\rb{t_k}}&\leq\frac{2^N}{T_\mathsf{min}^N\cdot N!}\cdot\underset{l=0,\dots,N}{\max}\vb{\NFD^{0}r\rb{k+l}}\\
    &=\frac{2^N}{T_\mathsf{min}^N\cdot N!}\cdot\underset{l=0,\dots,N}{\max}\vb{\frac{\rb{t_{k+l}-\tau}^N}{N!}g^{\rb{N-1}}\rb{\zeta_{k+l}}}\\
    &=\frac{2^N}{T_\mathsf{min}^N\cdot N!} \cdot \norm{g^{\rb{N-1}}}_\infty \underset{l=0,\dots,N}{\max}\vb{\frac{\rb{t_{k+l}-\tau}^N}{N!}}\\
    &\leq\frac{2^N}{T_\mathsf{min}^N\cdot N!}\cdot \frac{\norm{g^{\rb{N-1}}}_\infty}{N!}\cdot \rb{\frac{t_{k+N}-t_k}{2}}^N\\
    &\leq \rb{\frac{N\cdot T_\mathsf{max}}{T_\mathsf{min}}}^N \cdot \frac{\norm{g^{\rb{N-1}}}_\infty}{\rb{N!}^2}.
\end{split}
\end{align*}
\endgroup
The following two inequalities follow using the Stirling and Bern\v{s}te\u{\i}n  {(cf. pg. 116, \cite{Nikol:2012:B})} inequalities, respectively.
\begin{align*}
    \vb{\NFD^N G\rb{t_k}}&\leq \rb{\frac{e\cdot T_\mathsf{max}}{T_\mathsf{min}}}^N \cdot \frac{\norm{g^{\rb{N-1}}}_\infty}{N!}\\
    &\leq \frac{1}{\Omega N!}\rb{\frac{T_\mathsf{max}}{T_\mathsf{min}}\Omega e}^N  \norm{g}_\infty.    
\end{align*}
\end{proof}

To prove Lemma \ref{lem:km_kM}, we first provide a few properties of $\mu_{l}^N\sqb{k}$ and $\beta_{l}^N\sqb{k}$ in the following proposition.

\statepropsolid{0.95\textwidth}{
\begin{prop}[Properties of $\mu_{l}^N$ and $\beta_{l}^N$]
The following hold,
\begin{equation}
\label{eq:supp_mu_beta}
\mathrm{supp}\ \mu_{l}^N\sqb{k}=\cb{l-N+1,\dots,l-1} \quad \mbox{and} \quad \mathrm{supp}\ \beta_{l}^N\sqb{k}=\cb{l-N+2,\dots,l}. 
\end{equation}
For {all} $ k \in \mathbb{Z}, \forall N \geq 2$ it follows that
\begin{align}
\label{eq:signs_mu_beta}
    \begin{split}        
        \mathrm{sign} \rb{\mu_{l}^N\sqb{k}} \cdot \mathrm{sign} \rb{\mu_{l}^N\sqb{{k+1}}} &\leq 0, \\
        \mathrm{sign} \rb{\beta_{l}^N\sqb{k}} \cdot \mathrm{sign} \rb{\beta_{l}^N\sqb{{k+1}}} &\leq 0, \\
        \mathrm{sign} \rb{\mu_{l}^N\sqb{k}} \cdot \mathrm{sign} \rb{\beta_{l}^N\sqb{{k}}} &\leq 0.
    \end{split}
\end{align}
The following bounds hold true.
\begin{align}
\label{eq:edge_bound_mu_beta}
    \begin{split}
        \frac{1}{N!\hspace{0.1em}  T_\mathsf{max}^{N-1}}&\leq \vb{\beta_{l}^N\hspace{-0.1em}\sqb{k}} \leq \frac{1}{N!\hspace{0.1em} T_\mathsf{min}^{N-1}},\ k\hspace{-0.1em}\in\hspace{-0.1em}\cb{l\hspace{-0.1em}-\hspace{-0.1em}N\hspace{-0.1em}+\hspace{-0.1em}2,l},\\
        \frac{1}{N!\hspace{0.1em} T_\mathsf{max}^{N-1}}&\leq \vb{\mu_{l}^N\hspace{-0.1em}\sqb{k}} \leq \frac{1}{N!\hspace{0.1em}  T_\mathsf{min}^{N-1}},\  k\hspace{-0.1em}\in\hspace{-0.1em}\cb{l\hspace{-0.1em}-\hspace{-0.1em}N\hspace{-0.1em}+\hspace{-0.1em}1,l\hspace{-0.1em}-\hspace{-0.1em}1},
    \end{split}
\end{align}
\begin{equation}
\label{eq:mid_bound_mu}
    \vb{\mu_{l}^N\sqb{{l-N+2}}}\leq \frac{N-2}{N!\ T_\mathsf{min}^{N-1}},
\end{equation}
\begin{equation}
\label{eq:mid_bound_beta}
    \vb{\beta_{l}^N\sqb{{l-1}}}\leq \frac{N-2}{N!\ T_\mathsf{min}^{N-1}}.
\end{equation}
\end{prop}
}

\begin{proof}
Equation \eqref{eq:supp_mu_beta} is derived directly from the definitions of $\mu_{l}^N$ and $\beta_{l}^N$ in Proposition \ref{prop2:Eg_mu_beta}. Moreover, \eqref{eq:signs_mu_beta} can be shown directly via induction for $N\geq 2$. Bounds \eqref{eq:edge_bound_mu_beta} are shown recursively by using $n T_\mathsf{min}\leq t_{k+n}-t_k \leq n T_\mathsf{max}$.

For (\ref{eq:mid_bound_mu}), we compute recursively an upper bound for $m_N\triangleq\vb{\mu_{l}^N\sqb{{l-N+2}}}$, where $m_2=0$ \eqref{eq:supp_mu_beta} and, for $N>2$,
\begin{align*}
m_{N+1}=\frac{\vb{\mu_{l}^N\sqb{{l-N+2}}-\mu_{l}^N\sqb{{l-N+3}}}}{t_{l+2}-t_{l-N+1}} & \EQc{eq:signs_mu_beta}\frac{m_N+\vb{\mu_{l}^N\sqb{{l-N+1}}}}{t_{l+2}-t_{l-N+1}}\\
&\leq\frac{m_N+\vb{\mu_{l}^N\sqb{{l-N+1}}}}{\rb{N+1} T_\mathsf{min}} \\
&\EQc{eq:edge_bound_mu_beta}{\leq}\frac{m_N}{\rb{N+1} T_\mathsf{min}}+\frac{1}{\rb{N+1}!\hspace{0.1em} T_\mathsf{min}^N}.
\end{align*}
Via induction, we derive \eqref{eq:mid_bound_mu}. The same procedure is used for the derivation of \eqref{eq:mid_bound_beta}.
\end{proof}

\begin{proof}[Proof of Lemma \ref{lem:km_kM}]
Due to (\ref{eq:supp_Eg}, \ref{eq:bound_G}) it directly follows that
\begin{equation}
    \vb{\NFD^N X \rb{t_k}}=\vb{\NFD^N G \rb{t_k}-\NFD^N E_g \rb{t_k}}< \th^N\sqb{k},
\end{equation}
for {all} $k \notin \mathbb{S}_N\Rightarrow \KN\subseteq \mathbb{M}_N\subseteq\mathbb{S}_N$. Furthermore, we note that 
\begin{align}
\label{eq:lower_bound_Eg}
\begin{split}
    \vb{\NFD^N X \rb{t_k}}
    &\geq \vb{\NFD^N E_g \rb{t_k}} - \vb{\NFD^N G \rb{t_k}}\\
    &> \vb{\NFD^N E_g \rb{t_k}} - \th^N\sqb{k}, \forall k\in{\Z_+}.
\end{split}    
\end{align}

Specifically, we will identify the values of $k$ for which $\vb{\NFD^N E_g \rb{t_k}}\geq 2 \th^N\sqb{k}$, which are contained in $\mathbb{M}_N$ according to \eqref{eq:lower_bound_Eg}. This yields upper bounds for $k_m$ and lower bounds for $k_M$. With $$E_k^N\triangleq \frac{\vb{\NFD^N E_g\rb{t_k}}}{2\lambda_h},$$ the values $k$ of interest satisfy 
\begin{equation}
    E_k^N\geq\frac{\th^N\sqb{k}}{\lambda_h}=\min\cb{\phi_{0,k+N-1},\phi_{0,k+N-2},\phi_{1,k-1},\phi_{1,k}}.
\end{equation}
An important step in the proof is to show that
\begin{gather}
%\begin{gathered}
\max\cb{E_{l-N+1}^N,E_{K_1-N+2}^N}\geq \phi_{0,K_1},\label{eq:lower_bounds_E1}\\
\max\cb{E_{K_1-1}^N,E_{K_1}^N}\geq \phi_{1,K_1}.
\label{eq:lower_bounds_E2}        
%\end{gathered}
\end{gather}
To show \eqref{eq:lower_bounds_E1}, we denote by $f_1,f_2:\sqb{0,1} \rightarrow \mathbb{R}$ two functions 
\begin{align*}
    f_1\rb{x}&=\vb{\bar{\mu}_1 x+\bar{\beta}_1\rb{1-x}}\overset{\eqref{eq:supp_mu_beta}}{=}\vb{\bar{\mu}_1}x,\\
    f_2\rb{x}&=\vb{\bar{\mu}_2 x+\bar{\beta}_2\rb{1-x}}\overset{\eqref{eq:signs_mu_beta}}{=}\vb{\vb{\bar{\mu}_2} x-\vb{\bar{\beta}_2}\rb{1-x}},
\end{align*}
where $\bar{\mu}_i=\mu_{K_1}^N\sqb{K_1-N+i}, \bar{\beta}_i=\beta_{K_1}^N\sqb{K_1-N+i}, i=1,2$.
Then we have that  $E_{K_1-N+1}^N=f_1\rb{p}$ and $E_{K_1-N+2}^N=f_2\rb{p}$ \eqref{eq:mu_k_beta_k}, which gives
\begin{equation}
\label{eq:partial_bound}
    \max\cb{E_{K_1-N+1}^N,E_{K_1-N+2}^N}\geq\min_{x\in\sqb{0,1}} \rb{\max_{i=1,2} f_i\rb{x}}.
\end{equation}
Function $f_1\rb{x}$ is strictly increasing, whereas $f_2\rb{x}$ has one zero $f\rb{x^\ast}=0$, where $$x^\ast=\frac{\vb{\bar{\beta}_2}}{\vb{\bar{\mu}_2}+\vb{\bar{\beta}_2}}.$$ Then $f_2\rb{x}$ is strictly decreasing for $x\in\sqb{0,x^\ast}$, and thus, given that $f_1(0)=0$ and $f_2(0)>0$, it follows that there exists a unique point $$x^{\ast\ast}\in\left(0,x^\ast\right)$$ such that $f_1\rb{x^{\ast\ast}}=f_2\rb{x^{\ast\ast}}$. Indeed, when solving $f_1\rb{x}=f_2\rb{x}$ for $x$, there are two solutions 
$$
x_k=\frac{\vb{\bar{\beta}_2}}{\vb{\bar{\beta}_2}+\vb{\bar{\mu}_2}+\rb{-1}^k\vb{\bar{\mu}_1}}, \quad k=1,2,
$$ 
such that $x_2<x^\ast<x_1$, and thus 
$$x^{\ast\ast}=\frac{\vb{\bar{\beta}_2}}{\vb{\bar{\beta}_2}+\vb{\bar{\mu}_2}+\vb{\bar{\mu}_1}}.$$ It follows that $ \max_{i=1,2}f_i\rb{x}=f_2\rb{x}, x\in\sqb{0,x^{\ast\ast}}$ and $ \max_{i=1,2}f_i\rb{x}=f_1\rb{x}, x\in\sqb{x^{\ast\ast},x^\ast}$. Given that for $x>x^\ast$ $f_2\rb{x}$ is strictly increasing, then $f_2\rb{x}>f_2\rb{x^\ast}, \forall x>x^\ast$, and thus
\begin{equation*}
    \max_{i=1,2}f_i\rb{x}\geq \max_{\underset{x\in\sqb{0,x^\ast}}{i=1,2}}f_i\rb{x}\geq f_1\rb{x^{\ast\ast}}={\vb{\bar{\mu}_1}}x^{\ast\ast}=\phi_{0,K_1},
\end{equation*}
which yields \eqref{eq:lower_bounds_E1} via \eqref{eq:partial_bound}.

The bound \eqref{eq:lower_bounds_E2} is proven following the same steps as for \eqref{eq:lower_bounds_E1}. Specifically, we denote by $f_3,f_4:\sqb{0,1} \rightarrow \mathbb{R}$ 
\begin{align*}
    f_3\rb{x}&=\vb{\bar{\mu}_3 x+\bar{\beta}_3\rb{1-x}}\overset{\eqref{eq:signs_mu_beta}}{=}\vb{\vb{\bar{\mu}_3} x-\vb{\bar{\beta}_3}\rb{1-x}},\\
    f_4\rb{x}&=\vb{\bar{\mu}_4 x+\bar{\beta}_4\rb{1-x}}\overset{\eqref{eq:supp_mu_beta}}{=}\vb{\bar{\beta}_4}\rb{1-x},
\end{align*}
where $$\bar{\mu}_i=\mu_{K_1}^N\sqb{K_1-4+i}, \bar{\beta}_i=\beta_{K_1}^N\sqb{K_1-4+i}, i=3,4.$$ Using $y=1-x$ we get the same problem as before, which yields \eqref{eq:lower_bounds_E2} via direct calculation. Finally, (\ref{eq:lower_bounds_E1},\ref{eq:lower_bounds_E2}) together with \eqref{eq:lower_bound_Eg} imply that $k_m\leq K_1-N+2$ and $k_M\geq K_1-1$. This completes the proof of the lemma via $\KN\subseteq\mathbb{S}_N$.
\end{proof}

\begin{proof}[Proof of Theorem \ref{th:convergence_tau}]
To show that $\widetilde{s}_1=s_1$, we observe that $\NFD^N G\rb{t_{k_m}}<\th^N\sqb{k_m}$ and $\NFD^N X\rb{t_{k_m}}\geq\th^N\sqb{k_m}$ which implies that $\mathrm{sign} \sqb{\NFD^N X\rb{t_{k_m}}}=-\mathrm{sign} \sqb{\NFD^N E_g\rb{t_{k_m}}}$ and thus \eqref{eq:mu_k_beta_k}
\begin{equation}
    \label{eq:s1}
    \widetilde{s}_1=\mathrm{sign} \sqb{\mu_{K_1}^N\sqb{{k_m}} p + \beta_{K_1}^N\sqb{{k_m}}\rb{1-p}}\cdot s_1.
\end{equation}
\begin{enumerate}
    \item If $k_M-k_m=N-1$, then, according to Lemma \ref{lem:km_kM}, we know that $k_m=K_1-N+1, k_M=K_1$. Then, due to (\ref{eq:supp_mu_beta}), $$\mu_{K_1}^N\sqb{{k_m}} p + \beta_{K_1}^N\sqb{{k_m}}\rb{1-p}= \mu_{K_1}^N\sqb{{k_m}} p.$$ We know that
    $\mu_{K_1}^2\sqb{{K_1-1}}=\frac{1}{t_{K_1+1}-t_{K_1-1}} >0 $. It can be shown  via induction that $\mathrm{sign}\rb{\mu_{K_1}^N\sqb{{K_1-N+1}}}=1, \forall N\geq2$, and thus $\widetilde{s}_1=s_1$ via \eqref{eq:s1}. Moreover, $\widetilde{\tau}_1=\frac{t_{K_1}+t_{K_1+1}}{2}$ which satisfies \eqref{eq:tau1_err1} given that $\tau_1\in\sqb{t_{K_1},t_{K_1+1}}$.
\end{enumerate}

\begin{enumerate}[start=2]
    \item If $k_M-k_m=N-2$, via Lemma \ref{lem:km_kM} there are two options
    \begin{enumerate}
        \item $k_m=K_1-N+1, k_M=K_1-1$. Here \eqref{eq:tau1_err2} based on the same reasoning as in 1).
        \item $k_m=K_1-N+2, k_M=K_1$. Here we have that
        \begin{equation}
        \label{eq:signs_updated}
            \widetilde{s}_1={\mathrm{sign} \rb{f(p)}}\cdot s_1,
        \end{equation}
        where,
        \begin{itemize}
  \item $f\rb{p}=\bar{\mu}_2 p+\bar{\beta}_2\rb{1-p}$. 
  \item $\bar{\mu}_2=\mu_{K_1}^{N}\sqb{K_1-N+2}$, and,
  \item $\bar{\beta}_2=\beta_{K_1}^{N}\sqb{K_1-N+2}$.
\end{itemize}
We know that $\vb{\NFD^N {Z} \rb{t_{K_1-N+1}}}<\th^N\sqb{K_1-N+1}$, which implies that  $\vb{\NFD^N E_g\rb{t_{K_1-N+1}}}<2\th^N\sqb{K_1-N+1}$ via \eqref{eq:lower_bound_Eg}. This can be {rewritten} as
        \begin{align*}
            \frac{\vb{\NFD^N E_g\rb{t_{K_1-N+1}}}}{2\lambda_h}=\vb{\bar{\mu}_1}\cdot p&< \frac{\th^N\sqb{K_1-N+1}}{\lambda_h}\leq\phi_{0,K_1},
        \end{align*}      
        where $\bar{\mu}_1=\mu_{K_1}^{N}\sqb{K_1-N+1}$.
        Using \eqref{eq:def_phi01}, this gives us $p<x^{\ast\ast}$, where $x^{\ast\ast}=\tfrac{\vb{\bar{\beta}_2}}{\vb{\bar{\beta}_2}+\vb{\bar{\mu}_2}+\vb{\bar{\mu}_1}}$. Moreover, using that $\mathrm{sign}\rb{\bar{\mu}_1}=1$ from 1) and \eqref{eq:signs_mu_beta}, it can be derived that $f\rb{p}$ is a strictly decreasing line and $f\rb{x^{\ast\ast}}{>}0$. Given that $p<x^{\ast\ast}$ we get $f(p)>0$ which implies $\widetilde{s}_1=s_1$ via \eqref{eq:signs_updated}. Moreover, $\widetilde{\tau}_1=t_{k_M+1}=t_{K_1+1}$ which satisfies \eqref{eq:tau1_err2}.
        \end{enumerate}        
    \item If $k_M-k_m=N-3$, Lemma \ref{lem:km_kM} implies $k_m=K_1-N+2, k_M=K_1-1$. Then $\widetilde{s}_1=s_1$ via the same derivation as in 2) {b}). Moreover, $\widetilde{\tau}_1=\frac{t_{K_1}+t_{K_1+1}}{2}$ which satisfies \eqref{eq:tau1_err1}. 
\end{enumerate}
\end{proof}

\begin{proof}[Proof of Proposition \ref{prop:final_error_bound}]
Assuming the values $\Delta G\rb{t_k} = \rb{\locav g}_k$ known, $g$ can be recovered recursively as \cite{Florescu:2017:B,Lazar:2004:J}
\begin{align}
\begin{split}
    {g}_0\rb{t}&=\syntharg{\Delta {G}\rb{t_k}}\rb{t},\\
    {g}_{n+1}\rb{t}&={g}_{n}\rb{t}+{g}_{0}\rb{t}-\syntharg{\locav {g}_{n}}\rb{t}, n>0.
\end{split}    
\label{eq:ideal_local_av_rec}
\end{align}
Then, provided that $T_\mathsf{max}<\frac{\pi}{\Omega}$ the following hold \cite{Florescu:2015}
\begin{gather}
    \norm{g-g_n}_{L^2}\leqslant\rb{\frac{T_\mathsf{max}\Omega}{\pi}}^{n+1}\norm{g}_{L^2},\label{eq:boundgp}\\
    \norm{\mathcal{I}-\synth\locav}_{L^2}\leqslant\frac{T_\mathsf{max}\Omega}{\pi},
    \label{eq:op_bound}
\end{gather}
where $\mathcal{I}$ is the identity operator, and thus $\lim_{n\rightarrow\infty}g_n = g$.

When comparing \eqref{eq:ideal_local_av_rec} with \eqref{eq:local_av_rec}, the main difference is in the initial condition $g_0$, which is determined by samples $\Delta {G}\rb{t_k}$.
In the following, we evaluate the error in computing these samples. 
We begin by making two observations.

\emph{Observation 1.} As derived at the end of the Proof for Theorem \ref{th:convergence_tau},  $$\widetilde{\tau}_r\in\cb{t_{K_r+1},\frac{t_{K_r}+t_{K_r+1}}{2}},$$ and thus, via the definition of $K_r$, we get ${\tau}_r,\widetilde{\tau}_r\in\sqb{t_{K_r},t_{K_r+1}}$, i.e., there are no spike times in between the true and estimated fold. 

\emph{Observation 2.} Because  $N\frac{2\delta \Omega g_\infty}{b-\lambda}\leq h^*$ (condition $S_2$ in Lemma \ref{lem:sufficient_conditions}) we are guaranteed to have at least two spike times between each two consecutive folding times $\tau_r,\tau_{r+1}$.

Using \eqref{eq:expanded_Z} and that $X\rb{t_k}$ is known, we get that $$\vb{\Delta \widetilde{G}\rb{t_{k}}-\Delta G\rb{t_k}}=\vb{\Delta \widetilde{E}_g\rb{t_{k}}-\Delta E_g\rb{t_k}}.$$
We will derive a bound for the latter by considering two cases.
\begin{enumerate}
    \item $k\in \mathbb{Z}_+^* \setminus \cb{K_1,\dots,K_R}$. We first consider the trivial case where $k<K_1 \Rightarrow E_g\rb{t_k}=\widetilde{E}_g\rb{t_k}=0$ (Observation 1). If $k>K_1$ then $\exists r_0\in\cb{1,\dots,R}$ s.t. $k\in\cb{K_{r_0}+1,\dots,K_{r_0+1}-1}$ (Observation 2). Using the expression of $E_g$ we get that
    \begin{align}
        \Delta E_g \rb{t_k}&=\sum\limits_{r=1}^{r_0} 2\lambda_h s_r \Delta \rb{\ind_{\left[\tau_r,\infty\right)}\rb{t_k} \label{eq:Eg_expr} \rb{t_k-\tau_r}}\\
        &=\sum\limits_{r=1}^{r_0} 2\lambda_h s_r \rb{t_{k+1}-t_k}.
    \end{align}
    The last equality uses that there are no folds in $\sqb{t_k,t_{k+1}}$ given that $k\neq K_{r_0}$. Using the same derivation for $\Delta \widetilde{E}_g$ and that $s_r=\widetilde{s}_r, \forall r \in\cb{1,\dots,R}$, we get that $\Delta \widetilde{E}_g \rb{t_k}=\Delta E_g \rb{t_k}$.

    \item $k\in\cb{K_1,\dots,K_R}$. Let $k=K_{r_0}, r_0\in\cb{1,\dots,R}$. Due to Observation 1 we have $\tau_{r_0},\widetilde{\tau}_{r_0}\in\sqb{t_{K_{r_0}},t_{K_{r_0}+1}}$. Then, using that $\ind_{\left[\tau_{r_0},\infty\right)}\rb{t_{k}}=0$ in \eqref{eq:Eg_expr} we get
\begin{equation}
         \Delta E_g \rb{t_k}=\sum_{r=1}^{r_0-1} 2\lambda_h s_r \rb{t_{k+1}-t_k}+2\lambda_h s_{r_0} \rb{t_{k+1}-\tau_{r_0}}
    \end{equation}
    Therefore $\vb{\Delta E_g\rb{t_k}-\Delta \widetilde{E}_g\rb{t_k}}=2\lambda_h \vb{\tau_{r_0}-\widetilde{\tau}_{r_0}}$.
\end{enumerate}
By combining 1) and 2), we get:
\begin{equation}
    \vb{\Delta \widetilde{G}\rb{t_{k}}-\Delta G\rb{t_k}}\leq 2\lambda_h \max_r \vb{\tau_r-\widetilde{\tau}_r}.
    \label{eq:Gtk_err}
\end{equation}
    
To bound the estimation of $g_n$, we define $e_n\triangleq \widetilde{g}_n-g_n \Rightarrow e_0=\syntharg{\Delta \widetilde{G}\rb{t_k}-\Delta G \rb{t_k}}$ and thus
\begin{equation}
    e_0\rb{t}=\sum_{k\in\mathbb{K}_r} \sqb{\Delta \widetilde{G}\rb{t_k}-\Delta G \rb{t_k}} \cdot \mathrm{sinc}_\Omega \rb{t-s_k}.
\end{equation}
Given that $\vb{\mathbb{K}_r}\leq 2R$, via the triangle inequality and \eqref{eq:Gtk_err} we have that
\begin{align}
\begin{split}
    \norm{e_0}_{L^2}&\leq 2R \max_k\vb{\Delta \widetilde{G}\rb{t_k}-\Delta G \rb{t_k}} \cdot \norm{\mathrm{sinc}_\Omega}_{L^2}\\
    &\leq 4 \lambda_h R\cdot T_\mathsf{max} \sqrt{\frac{\Omega}{\pi}}.
\end{split}    
    \label{eq:bound_e0}
\end{align}
Furthermore, we have that (\ref{eq:local_av_rec}, \ref{eq:ideal_local_av_rec})
\begin{align}
    \norm{e_n}_{L^2}&=\norm{e_{n-1}+e_0-\synth\locav e_{n-1}}_{L^2}\\
    &\leq\norm{e_0}_{L^2}+\norm{\rb{\mathcal{I}-\synth\locav}e_{n-1}}_{L^2}\\
    &\overset{\eqref{eq:op_bound}}{\leq}\norm{e_0}_{L^2}+\norm{e_{n-1}}_{L^2}\cdot \frac{T_\mathsf{max}\Omega}{\pi}.  
\end{align}
Via induction, given that $T_\mathsf{max}<\frac{\pi}{\Omega}$, we get
\begin{align}
    \norm{e_n}_{L^2}&\leq \norm{e_0}_{L^2} \sum_{l=0}^n \rb{\frac{T_\mathsf{max}\Omega}{\pi}}^l =\norm{e_0}_{L^2}\frac{1-\rb{\frac{T_\mathsf{max}\Omega}{\pi}}^{n+1}}{1-\frac{T_\mathsf{max}\Omega}{\pi}}\leq \frac{\norm{e_0}_{L^2}}{1-\frac{T_\mathsf{max}\Omega}{\pi}}.\label{eq:bound_ep}
\end{align}
Finally, we bound the reconstruction error as
\begin{equation}
    \norm{\widetilde{g}_n-g}_{L^2}\leq \norm{\widetilde{g}_n-g_n}_{L^2}+\norm{g-g_n}_{L^2},
\end{equation}
which yields the desired result via (\ref{eq:bound_e0},\ref{eq:bound_ep},\ref{eq:boundgp}), where $T_\mathsf{max}$ satisfies \eqref{eq:ISIbounds}.
\end{proof}

\section{Conclusions and Future Work}
\label{sect:conclusions}

\bpara{Summary of Results.} We proposed a new encoding model based on the modulo event-driven sampling (MEDS) pipeline consisting of a modulo-hysteresis nonlinearity in series with an asynchronous sigma-delta modulator (ASDM) model. Our result contributes to the existing efforts of alleviating the dynamic range restrictions of existing EDS sensors. We provide  mathematical guarantees for the input reconstruction, and validate the MEDS model with synthetic and hardware experiments. Numerical simulations show that the proposed method can recover inputs with a dynamic range up to $20$ times the modulo threshold.

\bpara{Future Work.} We plan on following up the results in several directions,
\begin{enumerate}
\item  The method can be extended to modulo models in series with other event-driven encoders, such as level-crossing encoders and biphasic integrate-and-fire neurons.
\item The nonuniform finite difference could be replaced by a more general operator for achieving a better recovery.
\item The current methodology can be extended to inputs in multiple dimensions.
\end{enumerate}

New exciting hardware applications such as event-driven cameras show the promise that this encoding scheme holds in the field of signal processing. Despite this, the sensor dynamic range restriction remains a fundamental bottleneck in the field. Recently, it was shown that a high dynamic range image can be recovered from uniform samples in a single capture \cite{Bhandari:2020:C}. Given that event-driven signals have an inherently low dynamic range, this line of research represents a natural step in the field that could lead to promising new encoding strategies.

\section*{Acknowledgements}
 {The authors thank AB's undergraduate mentee Jerry Chen for help with the experiments, specially during the pandemic and the anonymous reviewers whose comments have improved the readability of our work.}

\ifCLASSOPTIONcaptionsoff
\newpage
\fi

\bibliographystyle{IEEEtran}

% Generated by IEEEtran.bst, version: 1.14 (2015/08/26)

\end{document}